\documentclass{article}
\usepackage[margin=1in]{geometry}
\usepackage{amsmath,amssymb}
\usepackage{amsfonts}
\usepackage{natbib}
\usepackage{psfrag}
\usepackage{enumerate}
\usepackage{graphicx}
\usepackage{amsthm}
\usepackage{xcolor}
\usepackage{bm}
\usepackage{xr}

\title{Modelling across extremal dependence classes}
\author{J.\ L.\ Wadsworth$^{1}$, J.\ A.\ Tawn$^{1}$, A.\ C.\ Davison$^{2}$ and D.\ M.\ Elton$^{1}$\\
$^{1}$Lancaster University, UK\\
$^{2}$Ecole Polytechnique F\'{e}d\'{e}rale de Lausanne, Switzerland}

\newcommand{\E}{\textsf{E}}
\newcommand{\Prob}{\textsf{P}}
\newcommand{\dsp}{\,\mbox{d}}
\newcommand{\abs}[1]{\lvert{#1}\rvert}

\newcommand{\norm}[2][m]{\lVert{#2}\rVert_{#1}}

\newcommand{\nf}{\tau}
\newcommand{\mnf}{\nu}
\newcommand{\muI}{\Omega_0}
\newcommand{\nuI}{\Omega}
\newcommand{\lep}{v'}
\newcommand{\rep}{v''}
\newcommand{\upm}{m_+}
\newcommand{\lowm}{m_-}
\newcommand{\upn}{\mu_1}
\newcommand{\lown}{\mu_2}
\newcommand{\xmarg}[1][]{\phi_{#1}}
\newcommand{\ixmarg}{l_A}
\newcommand{\xq}{q_A(t^{\beta})}
\newcommand{\iymarg}{l_B}
\newcommand{\yq}{q_B(t^{\gamma})}
\newcommand{\xyj}[1][]{\theta_{\beta,\gamma}{#1}}
\newcommand{\jsvf}{g_v}
\newcommand{\cm}{C_\delta}
\newcommand{\esupp}{\Lambda}
\newcommand{\intg}[1]{\mathcal{I}_{#1}}
\newcommand{\Isvi}{\phi}
\newcommand{\Isvo}{u}
\newcommand{\Iind}{\beta}
\newcommand{\Lrvi}{u}
\newcommand{\Lrvo}[1][]{\phi_{#1}}
\newcommand{\Lcrvo}[1][]{\psi_{#1}} 
\newcommand{\Lind}{a}
\newcommand{\LimA}{\alpha}
\newcommand{\Lintv}[1][]{I_{#1}} 
\newcommand{\Lii}{\rho}

    \newcommand\independent{\protect\mathpalette{\protect\independenT}{\perp}}
    \def\independenT#1#2{\mathrel{\rlap{$#1#2$}\mkern2mu{#1#2}}}

\newtheorem{lem}{Lemma}
\newtheorem{prop}{Proposition}
\theoremstyle{definition}
\newtheorem{Example}{Example}
\theoremstyle{definition}
\newtheorem{Case}{Case}
\theoremstyle{definition}
\newtheorem{Remark}{Remark}
\theoremstyle{definition}
\newtheorem{Cond}{Condition}
\newtheorem{As}{Assumption}

\numberwithin{equation}{section}

\begin{document}
\maketitle

\begin{abstract}
Different dependence scenarios can arise in multivariate extremes, entailing careful selection of an appropriate class of models. In bivariate extremes, the variables are either asymptotically dependent or are asymptotically independent. Most available statistical models suit one or other of these cases, but not both, resulting in a stage in the inference that is unaccounted for, but can substantially impact subsequent extrapolation. Existing modelling solutions to this problem are either applicable only on sub-domains, or appeal to multiple limit theories. We introduce a unified representation for bivariate extremes that encompasses a wide variety of dependence scenarios, and applies when at least one variable is large. Our representation motivates a parametric model that encompasses both dependence classes. We implement a simple version of this model, and show that it performs well in a range of settings.
\end{abstract}

\textbf{Keywords:} asymptotic independence, censored likelihood, conditional extremes, dependence modelling, extreme value theory, multivariate regular variation.

\section{Introduction}
\label{sec:Introduction}

The first challenge faced when modelling extremes of two or more variables is to decide which type of dependence they exhibit.  There are two possibilities in the bivariate case. For a random vector $(Z_1,Z_2)$, with marginal distributions $F_1,F_2$, define the limiting probability
\begin{align}
\chi = \lim_{u\to 1}\Prob\{F_1(Z_1)>u~|~F_2(Z_2)>u\},  \label{eq:chi}
\end{align}
if it exists. The pair $(Z_1,Z_2)$ are termed \emph{asymptotically dependent} if $\chi>0$, and \emph{asymptotically independent} if $\chi=0$. In higher dimensions the situation becomes more complicated;  \citet{WadsworthTawn13} outline the idea of $k$-dimensional joint tail dependence, which is summarized by $\sum_{i=0}^{k-2} {k\choose i}$ limits such as~\eqref{eq:chi}. For this reason, we focus on bivariate data, but discuss higher dimensional cases in Section~\ref{sec:Discussion}.

It is important to detect the appropriate dependence class because most models for bivariate extremes encompass one or the other, but not both. Classical multivariate extreme value theory \citep[e.g.,][Chapter 5]{Resnick87} yields asymptotic dependence models \citep{ColesTawn91,deHaandeRonde98}. Its first stage is usually to transform variables to a common marginal distribution. Suppose that $(X_P,Y_P) = [\{1-F_1(Z_1)\}^{-1},\{1-F_2(Z_2)\}^{-1}]$ have marginal standard Pareto distributions (interpreted asymptotically, if $F_1,F_2$ are discontinuous). In the asymptotic dependence case the basic modelling principle is that for an arbitrary pair of norms $\|\cdot\|_a$ and $\|\cdot\|_b$, the pseudo angular and radial variables 
\begin{align}
\bm{W} = (X_P,Y_P) / \|(X_P,Y_P)\|_a,\ R = \|(X_P,Y_P)\|_b, \label{eq:RW}
\end{align}
become independent in the limit, in the sense that
\begin{align}
\lim_{t\to\infty} \Prob\{\bm{W} \in B, R>t(r+1) ~|~ R>t\} = H(B)(r+1)^{-1},\ r\geq 0,~ B\subset\mathcal{S}^a:=\{\bm{w}\in\mathbb{R}^2_+ : \|\bm{w}\|_a=1\}, \label{eq:RWAD}
\end{align}
for continuity sets of the limit measure $H$. The limit holds for both dependence classes, but is only useful under asymptotic dependence: under any form of asymptotic independence, $H(\cdot)$ is a discrete two-point distribution that places atoms of probability on the endpoints of the continuous arc $\mathcal{S}^a$, $(0,1)/\|(0,1)\|_a$, $(1,0)/\|(1,0)\|_a$. Since $\|\cdot\|_a$ is arbitrary, we henceforth use the $L_1$-norm, $\|\cdot\|_1$, and redefine $H$ to be the limiting distribution of $W=X_P/(X_P+Y_P)$, with $H(w)=H([0,w])$ $(0\leq w\leq 1)$. Under asymptotic dependence, $H$ has mass on the interior of $[0,1]$ and likelihood-based statistical modelling typically assumes the existence of a \emph{spectral density},\ $h(w)=\dsp H(w)/\dsp w$ \citep{ColesTawn91}. One common goal of multivariate extreme value modelling is to estimate probabilities such as $\Prob\{(Z_1,Z_2)\in A\}$, where the set $A$ is extreme in at least one margin. Under asymptotic dependence, this  is aided by inference on $h$, and the independent limit distribution of the scaling appearing in~\eqref{eq:RWAD}.

The degeneracy of $H$ under asymptotic independence occurs because~\eqref{eq:chi} implies that the very largest values of $Z_1$ or $Z_2$, and hence of $X_P$ or $Y_P$, occur singly, pushing all the mass of $W$ to the boundaries of the interval $[0,1]$. This is due to the heavy tails of Pareto random variables: since the high quantiles on the Pareto scale are very large, one of $X_P$ and $Y_P$ will dominate the other when $R$ is extreme.

This argument suggests that the choice of margins is central to simplifying extremal dependence modelling. Thus, rather than~\eqref{eq:RWAD}, we assume that there exist a common marginal distribution $F:(0,x^F)\to[0,1]$, where $x^F \leq \infty$ is the upper endpoint of the support, a norm $\|\cdot\|_*$, and normalization functions $a(t)>0$ and $b(t)$, such that the positive random variables $(X,Y) = [F^{-1}\{F_1(Z_1)\},F^{-1}\{F_2(Z_2)\}]$ satisfy
\begin{align}
\lim_{t\to\infty} \Prob\left\{ \left. \frac{X}{X+Y}\leq w, \|(X,Y)\|_*> a(t)r + b(t) \right|  \|(X,Y)\|_*>  b(t) \right\} = J(w)\bar{K}(r),\ r\geq 0, \label{eq:RWgen}
\end{align}
at continuity points of $J$, where $J$ is a non-degenerate probability distribution having mass on the interior of $[0,1]$, and $\bar{K}$ is the survivor function of the generalized Pareto, GP$(\sigma,\lambda)$, distribution. That is, 
\begin{align}
\bar{K}(r) = (1+\lambda r/\sigma)_+^{-1/\lambda},\ r\geq0, ~\sigma>0,~ \lambda\in\mathbb{R},~a_+=\max(a,0); \label{eq:gp}
\end{align}
the case $\lambda=0$ is interpreted as the limit $\bar{K}(r)=\exp(-r/\sigma)$. In~\eqref{eq:RWgen}, $a(t)$ and $b(t)$ are the same as in the theory for univariate extremes for the variable $\|(X,Y)\|_*$; see Chapter~1 of \citet{Leadbetter83}, for example. When $(Z_1, Z_2)$ are asymptotically dependent and $F(\cdot) = 1-(\cdot)^{-1}$, so that $(X,Y)$ have standard Pareto margins, then \eqref{eq:RWgen} is equivalent to~\eqref{eq:RWAD}, with $a(t)=b(t)=t$ and $\bar{K}(r) = (1+r)^{-1}$; thus $\sigma=\lambda=1$, and the distribution $J$ in \eqref{eq:RWgen} equals $H$ as defined following~\eqref{eq:RWAD}. When $(Z_1,Z_2)$ are asymptotically independent, then a marginal $F$ with a lighter tail is required to obtain a distribution $J$ placing mass in $(0,1)$. The extremal dependence is then described by the combination of $J$, $\|\cdot\|_*$ and $\lambda$. Section~\ref{sec:Motivation} contains further discussion of the meaning and interpretation of~\eqref{eq:RWgen}, and motivates it with a variety of examples.

Under asymptotic dependence, the norms used in transformation~\eqref{eq:RW} to $\bm{W}$ and $R$ are arbitrary and need not be the same. In~\eqref{eq:RWgen}, we have again defined a pseudo angular and radial transformation
\begin{align}
W = X /(X+Y),\ R = \|(X,Y)\|_*, \label{eq:RW2}
\end{align}
where for later simplicity we use the $L_1$-norm in the definition of $W$, but the norm $\|\cdot\|_*$ defining $R$ must be chosen so that the limit~\eqref{eq:RWgen}  holds. The inverse of~\eqref{eq:RW2} is
\begin{align}
(X,Y) = R \left(\frac{W}{\|(W,1-W)\|_*},\frac{1-W}{\|(W,1-W)\|_*}\right). \label{eq:XYRW}
\end{align}
When assumption~\eqref{eq:RWgen} holds, we see from~\eqref{eq:XYRW} that for large $R$ the variables $(X,Y)$ behave as if the angular component $(W/\|(W,1-W)\|_*, (1-W)/\|(W,1-W)\|_*)$ is randomly scaled by an independent generalized Pareto variable. However, it is not straightforward to exploit this statistically, because the flexibility in~\eqref{eq:RWgen} stems from not having specified the margins $F$ in which we make the pseudo radial-angular transformation. Nonetheless, the dependence structure defined by~\eqref{eq:XYRW} must describe a rich variety of extremal dependencies, and motivates a copula model, described in Section~\ref{sec:Model}, that we can apply to both asymptotically dependent and asymptotically independent data. This model can indeed capture many extremal dependence structures, reproducing the entire ranges of common summary statistics for extremal dependence in both dependence classes.

In Section~\ref{sec:LinkExisting} we review current statistical methods for bivariate extremes, focussing on those providing a non-trivial treatment of asymptotic independence. In Section~\ref{sec:Motivation} we present examples to illustrate assumption~\eqref{eq:RWgen}, and discuss further the interpretation of the limit assumption. Section~\ref{sec:Model} introduces a statistical model and describes its dependence properties. Inference approaches are developed in Section~\ref{sec:Inference}, with some simulations to assess how well a given version of the model can estimate rare event probabilities, and in Section~\ref{sec:Applications} we apply our model to oceanographic data previously analyzed using both dependence structures. We conclude the article by outlining extensions to higher dimensions and discussing related issues.

\section{Existing methodology incorporating asymptotic independence}
\label{sec:LinkExisting}

Many inferential approaches for extremal dependence assume the applicability of equation~\eqref{eq:RWAD} with asymptotic dependence; see for example \citet{ColesTawn91}, \citet{Einmahl97}, \citet{deHaandeRonde98}, \citet{Mikosh05} and \citet{SabourinNaveau14}. \citet{LedfordTawn97} noted a gap in the theory for practical treatment of asymptotic independence and introduced the \emph{coefficient of tail dependence}, $\eta\in(0,1]$. For $(X_P,Y_P)$ as defined in Section~\ref{sec:Introduction}, this coefficient may be defined through the equation
\begin{align}
\Prob(X_P>tx,Y_P>ty) =\mathcal{L}(tx,ty)t^{-1/\eta}(xy)^{-1/2\eta},\ tx,ty \geq1, \label{eq:eta}
\end{align}
where $\mathcal{L}$ is bivariate slowly varying at infinity, i.e., $\mathcal{L}(tx,ty)/\mathcal{L}(t,t)\to d\{x/(x+y)\}$, $t\to\infty$, with $d:(0,1)\to (0,\infty)$ termed the \emph{ray dependence function}, depending only on the \emph{ray} $q:=x/(x+y)$. When $\eta=1$ and $\mathcal{L}(t,t)\not\to 0$ as $t\to\infty$ we obtain asymptotic dependence, but otherwise there is asymptotic independence. 

Setting $x=y=1$ in~\eqref{eq:eta} gives $\Prob(X_P>t,Y_P>t) = \mathcal{L}(t,t)t^{-1/\eta}$. Under asymptotic dependence, $\eta=1$ and the dependence is summarized by the parameter $\chi=\lim_{t\to\infty}\mathcal{L}(t,t)>0$. Under asymptotic independence, $\chi=0$ and $\eta \leq 1$ summarizes the degree of dependence.

The parameters $\chi$ and $\eta$ do not explain all the features of the extremal dependence of $(Z_1,Z_2)$. Under asymptotic dependence, the function $d(q)$ prescribes how to scale $(xy)^{-1/2}$ in order to find joint survivor probabilities across different rays, $q \in [0,1]$ in Pareto margins. When $\chi>0$, the link between $d$ and $H$, as defined following equation~\eqref{eq:RWAD}, is 
\begin{align}
d(q) = \frac{2}{\chi}\int_{0}^1\min\left\{w\left(\frac{1-q}{q}\right)^{1/2}, (1-w)\left(\frac{q}{1-q}\right)^{1/2}\right\} \dsp H(w).\label{eq:dH}
\end{align}
By definition, $d(1/2)=1$, so $\chi=2 \int_{0}^{1} \min(w,1-w) \dsp H(w)$. \citet{RamosLedford09} offered a characterization of the function $d(q)$ when $\eta\neq 1$, beginning with the limit assumption
\begin{align}
\lim_{t\to\infty}\Prob(X_P > tx, Y_P> ty~|~X_P > t, Y_P> t) = d\{x/(x+y)\}(xy)^{-1/2\eta},~~~x,y \geq 1. \label{eq:RL}
\end{align}
In this case we may write 
\begin{align}
 d(q) = \eta \int_{0}^{1} \min\left\{w\left(\frac{1-q}{q}\right)^{1/2}, (1-w)\left(\frac{q}{1-q}\right)^{1/2}\right\}^{1/\eta} \dsp H_\eta(w),
 \label{eq:dHeta}
\end{align}
where $H_\eta$ is the \emph{hidden angular measure}, characterized in \citet{RamosLedford09}; see also \citet{Resnick02, Resnick06} and \citet{DasResnick14} for further details of this framework of \emph{hidden regular variation}. Suitable parametric models for $H_\eta$ give probability models for simultaneously extreme random variables on regions of the form $(X_P,Y_P)\in (v,\infty)^2$ for large $v$; see \citet{RamosLedford09} for examples. 

Unfortunately the Ramos--Ledford--Tawn approach is applicable only within regions where both variables are large. However, under asymptotic independence, the variables $(X_P,Y_P)$ do not grow in their joint extremes at the same rate as their marginal extremes, so these may not be the regions of most practical interest. \citet{WadsworthTawn13} provided an alternative representation for multivariate tail probabilities, allowing study of regions where one variable may be larger than the other. Their assumption was
\begin{align}
\Prob(X_P > t^\beta, Y_P> t^\gamma) = L(t;\beta,\gamma) t^{-\kappa(\beta,\gamma)}, ~~~\beta,\gamma\geq 0, \max(\beta,\gamma)>0, \label{eq:WT}
\end{align}
where the function $\kappa$ is homogeneous of order 1, and the function $L(\cdot;\beta,\gamma)$ is slowly varying at infinity, i.e., for all $a>0$, $\lim_{t\to\infty} L(t a;\beta,\gamma)/L(t;\beta,\gamma) = 1$. Under asymptotic independence $\kappa$ was shown to display structure similar to that provided by $d$ under asymptotic dependence. Representation~\eqref{eq:WT} is useful for estimation of joint survivor probabilities when one variable may be much larger than the other, although the inferential methodology of \citet{WadsworthTawn13} does not easily extend to regions more general than joint survivor regions. Example~\ref{eg:kappa} in Section~\ref{sec:Motivation} covers some special cases of this set-up.

\citet{HeffernanTawn04} developed a very general modelling assumption that we present in the adapted form of \citet{HeffernanResnick07}. For $(X_E,Y_E) = (-\log\{1-F_1(Z_1)\},-\log\{1-F_2(Z_2)\})$ with (asymptotically) standard exponential marginal distributions, they assume the existence of a non-degenerate $G$~in
\begin{align}
\lim_{t\to\infty} \Prob\left\{\left.\frac{X_E-b(Y_E)}{a(Y_E)}\leq x, Y_{E}>t+y ~\right |~ Y_E>t\right\} = G(x)e^{-y},~~ y\geq 0.\label{eq:HT}
\end{align}
Inference under~\eqref{eq:HT} is semiparametric, as the functions $a(Y_E)$ and $ b(Y_E)$ are typically chosen to be $Y_E^{\alpha}, \beta Y_E$, $\alpha\in(-\infty,1)$, $\beta\in[0,1]$, for non-negative dependence, and $G$ is estimated nonparametrically. Asymptotic dependence arises in the model only when $\alpha=0$, $\beta=1$, and then any structure is captured through $G$. Once more the limiting independence of the normalized $Y_E$ and $\{X_E-b(Y_E)\} / a(Y_E)$ is crucial to the inference. This method is a very flexible approaches to multivariate extreme value modelling, though we address some of its drawbacks with the representation~\eqref{eq:RWgen} and the associated model to be developed in Section~\ref{sec:Model}. One problem is that when conditioning on different variables, consistency of the resulting models is an unresolved issue \citep{LiuTawn14}. The need for nonparametric estimation of $G$ may be viewed as a strength or weakness, but can lead to difficulties in estimating non-zero probabilities \citep{PengQi04,WadsworthTawn13}. 

Like the methods described above, the new approach described in Section~\ref{sec:Model} is suitable for both asymptotically dependent and asymptotically independent data. However, it is motivated by a single limit representation, and may be applied when either variable is large. Moreover, our framework allows a smooth transition across the dependence class boundary, in a sense to be described in Section~\ref{sec:Transition}.

\section{Limit Assumption}
\label{sec:Motivation}

In Section~\ref{sec:Equivalent} we provide a condition that is equivalent to~\eqref{eq:RWgen} under additional smoothness assumptions. This condition is useful to illustrate applicability of~\eqref{eq:RWgen} when these extra assumptions are met. In Section~\ref{sec:Uniqueness} we discuss flexibility in how the limit may be exploited, and then discuss the interpretation of the limit assumption. 
A variety of examples are presented in Section~\ref{sec:RandomScaling}.

\subsection{Alternative Condition}
\label{sec:Equivalent}

Suppose that $(X,Y) = [F^{-1}\{F_1(Z_1)\},F^{-1}\{F_2(Z_2)\}]$ are continuous random variables with a joint density, so this is also true for $(R,W)$, as defined in~\eqref{eq:RW2}. This assumption is more restrictive than necessary, but it facilitates development and is often reasonable. Let  $c(u_1,u_2)$ denote the density of the \emph{copula}, i.e., the density of $\{F(X),F(Y)\} = \{F_1(Z_1),F_2(Z_2)\}$. Then, with $f$ denoting the density of $F$, the joint density of $(X,Y)$ is $f_{X,Y}(x,y) = c\{F(x),F(y)\}f(x)f(y)$. The Jacobian of the transformation from $(X,Y)$ to $(R,W)$ as defined in~\eqref{eq:RW2} is $r\|(w,1-w)\|_*^{-2}$, and the density $f_{R,W}(r,w)$ of $(R,W)$ equals
\begin{align}
  c\left[F\left\{\frac{rw}{\|(w,1-w)\|_*}\right\},F\left\{\frac{r(1-w)}{\|(w,1-w)\|_*}\right\}\right]f\left\{\frac{rw}{\|(w,1-w)\|_*}\right\}f\left\{\frac{r(1-w)}{\|(w,1-w)\|_*}\right\} \frac{r}{\|(w,1-w)\|_*^2}. \label{eq:rwcopula}
\end{align}

To demonstrate applicability of~\eqref{eq:RWgen}, we use the following simpler condition, which is valid when the relevant densities and limits exist. In Appendix~\ref{sec:AuxiliaryResults} we show that under mild assumptions~\eqref{eq:RWgen} is implied by
 \begin{align}
\lim_{t\to\infty} \Prob\{W\leq w \mid R = b(t)\} = J(w), \label{eq:dfcondn}
\end{align}
with $b(t) = F_{R}^{-1}(1-1/t)$, the $1-1/t$ quantile of $R$; or, terms of the joint density function $f_{R,W}(r,w)$, 
\begin{align}
\int_{0}^{w} f_{R,W}\{b(t),v\} \,\mbox{d}v \sim J(w)\int_{0}^{1} f_{R,W}\{b(t),v\} \,\mbox{d}v,~~ t\to\infty. \label{eq:denscondn}
\end{align}
Thus, when integration over the $W$ coordinate does not affect the rate at which the joint density decays in $r$ as $r\to r^F := \sup\{r:F_R(r)<1\}$, then condition~\eqref{eq:dfcondn}, and hence~\eqref{eq:RWgen}, is satisfied. Expression~\eqref{eq:rwcopula} shows how the transformed margins, defined by $F,f$, and the copula, $c$, interact for~\eqref{eq:denscondn} to apply.

In order to study the domain of attraction of the radial variable $R$, we assume differentiability of its density $f_{R}(r)$, and define the reciprocal hazard function $h_R(r):=\{1-F_{R}(r)\}/f_R(r)$. If $\lim_{r\to\infty}h_R'(r) =:\lambda \in (-\infty,\infty)$ then $R$ lies in the domain of attraction of the GP distribution with shape parameter $\lambda$ \citep{Pickands86}. Moreover if one takes $b(t) = F_{R}^{-1}(1-1/t)$, and $a(t)=h_R\{b(t)\}$, then $\sigma=1$ in~\eqref{eq:gp}, i.e.
\begin{align*}
\lim_{t\to\infty} \frac{1-F_R\{a(t)r + b(t)\}}{1-F_{R}\{b(t)\}} = \bar{K}(r) = (1+\lambda r)_+^{-1/\lambda}, ~~ r\geq 0.
\end{align*}

\subsection{Uniqueness of limits}
\label{sec:Uniqueness}
In general, for a given copula, no unique choice of marginal distribution $F$  leads to assumption~\eqref{eq:RWgen} being satisfied. Consider, for example, the independence copula, with $c(u_1,u_2)=1, (u_1,u_2)\in[0,1]^2$. The following cases are all covered by~\eqref{eq:RWgen}:

\begin{enumerate}[(i)]
 \item gamma margins, with shape parameter $\alpha>0$.  Then $R = \|(X_G,Y_G)\|_* = X_G+Y_G$, has a GP$(1,0)$ limit. The limiting distribution for $W$ is Beta$(\alpha,\alpha)$;
 \item Weibull margins, with shape parameter $\alpha>1$. Then $R = \|(X_W,Y_W)\|_* = (X_W^\alpha+Y_W^\alpha)^{1/\alpha}$, has a GP$(1,0)$ limit. The limiting distribution for $W$ has density $j(w) \propto w^{\alpha-1}(1-w)^{\alpha-1}\{w^{\alpha}+(1-w)^{\alpha}\}^{-2}$;
 \item uniform$(0,1)$ margins. Then $R=\|(X_U,Y_U)\|_* = \max(X_U,Y_U)$, has a GP$(1,-1)$ limit. The limiting distribution for $W$ has density $j(w) \propto \max(w,1-w)^{-2}$; 
 \item truncated Gaussian margins. Then $R=\|(X_N,Y_N)\|_* = (X_N^2+Y_N^2)^{1/2}$, has a GP$(1,0)$ limit. The limiting distribution for $W$ has density $j(w) \propto \{w^2 + (1-w)^2\}^{-1}$.
\end{enumerate}
The corresponding marginal densities may all be expressed as $f(x) = x^\beta e^{-x^\gamma} \gamma/\Gamma\{(\beta+1)/\gamma\}$, with (i) $\beta=\alpha-1,\gamma=1$; (ii) $\beta=\alpha-1, \gamma=\alpha$; (iii) $\beta=0, \gamma \to \infty$; and (iv) $\beta=0,\gamma=2$. In each case the norm $\|\cdot\|_*$ is the $L_\gamma$ norm, and the resulting density for $W$ satisfies 
$$
j(w)\propto w^{\beta}(1-w)^\beta/\{w^\gamma+(1-w)^\gamma\}^{(2\beta+2)/\gamma} =  w^{\beta}(1-w)^\beta/\|(w,1-w)\|_*^{2\beta+2},\quad 0<w<1, 
$$
 demonstrating a link between the margins of $(X,Y)$, the norm $\|\cdot\|_*$, and the distribution $J(w)$.

This lack of uniqueness also applies to multivariate regularly varying random vectors with asymptotically dependent copulas: equal heavy-tailed margins with any positive shape parameter will give a convergence as in~\eqref{eq:RWgen}, and the resulting distribution of $W$ will depend on this shape parameter and the norm used to define $R$; see Example~\ref{eg:mvrv} of Section~\ref{sec:RandomScaling}. Hence in considering how the distribution $J$ describes the extremal dependence, one must simultaneously consider $\lambda$, $\|\cdot\|_*$ and $J$. In convergence~\eqref{eq:RWAD}, by contrast, the effect of the margins is removed by standardization, and the extremal dependence depends only on $H$ and the norm used to define $R$.

The necessity of considering $\lambda$, $\|\cdot\|_*$ and $J$ together can be more clearly seen by observing what convergence~\eqref{eq:RWgen} implies for that of the normalized $(X,Y)$. Multiplying $\{\|(X,Y)\|_*-b(t)\}/a(t)$ by $(X,Y)/\|(X,Y)\|_*$, and conditioning on the event ${\mathcal B}=\{\|(X,Y)\|_* > b(t)\}$, the continuous mapping theorem gives that on ${\mathcal B}$, 
\begin{align}
\frac{(X,Y) }{a(t)} - \frac{b(t)}{a(t)}\frac{(X,Y)}{\|(X,Y)\|_*}  ~ \overset{d}{\to} ~R^*(W_1,W_2),\quad  t\to\infty, \label{eq:xylim}
\end{align}
with $W_1=W/\|(W,1-W)\|_*$, $W_2=(1-W)/\|(W,1-W)\|_*$. Here $\overset{d}{\to}$ denotes convergence in distribution, and $R^*\sim$ GP$(1,\lambda)$ is a random variable with survivor function $\bar{K}$. Equations~\eqref{eq:RWgen}, \eqref{eq:XYRW} and~\eqref{eq:xylim} suggest that for large $t$ we have the approximate distributional equality on ${\mathcal B}$, 
\begin{align}
(X,Y) \overset{d}{\approx} \{a(t)R^* +b(t)\} (W_1,W_2).~\label{eq:modapprox}
\end{align}
 Therefore the extremes of   $(X,Y)$ are described by the combination of the shape parameter $\lambda$, the norm $\|\cdot\|_*$ defining the sphere $\mathcal{S}^*$ on which $(W_1,W_2)$ live, and the distribution $J$ giving the density of $W$ on $[0,1]$.

\subsection{Examples}
\label{sec:RandomScaling}
We present three broad classes of examples, assuming throughout that derivatives of second order terms are also second order.

\begin{Example}
\label{eg:mvrv}
Suppose that $(X,Y)$ have $\alpha$-Pareto margins, $\Prob(X>x)=x^{-\alpha}, x>1$, and that $\Prob(X>tx,Y>ty)$ is a differentiable bivariate regularly varying function of index $-\alpha$ as $t\to\infty$. Then one can write
\begin{align*}
\Prob(X>tx,Y>ty) = \{1+o(1)\}\delta^{(\alpha)}(tx,ty) =\{\chi +o(1)\} d^{(\alpha)}\{x/(x+y)\}(xy)^{-\alpha/2} t^{-\alpha}, ~~t\to\infty,\quad  tx,ty>1,
\end{align*}
with $\chi>0$ as in~\eqref{eq:chi}, $\delta^{(\alpha)}$ a homogeneous function of order $-\alpha$, and $d^{(\alpha)}$ the associated ray dependence function, discussed in Section~\ref{sec:LinkExisting}. Such examples are asymptotically dependent. Then taking $\|\cdot\|_*=\|\cdot\|$, an arbitrary  norm, yields
\begin{align*}
f_{R,W}(r,w) = \{1+o(1)\}r^{-1-\alpha}\delta_{12}^{(\alpha)}\left\{\frac{w}{\|(w,1-w)\|},\frac{1-w}{\|(w,1-w)\|}\right\}\frac{1}{\|(w,1-w)\|^2},\quad r\to\infty;
\end{align*}
here $\delta_{12}^{(\alpha)}$, the joint derivative of $\delta^{(\alpha)}$, is homogeneous of order $-\alpha-2$. The reciprocal hazard function of $R$ satisfies $h_{R}(r)=r\{1/\alpha+o(1)\}$ ($r\to\infty$), so the limiting distribution of normalized exceedances of $R$ is generalized Pareto with $\lambda=1/\alpha$. The limiting density of $W$ is 
\begin{align*}
j(w) \propto \delta_{12}^{(\alpha)}\left(w,1-w\right)\|(w,1-w)\|^\alpha,\quad w\in(0,1).
\end{align*}
\end{Example}

\begin{Example}
\label{eg:kappa}
Suppose that  $(X,Y)$ have standard exponential margins, and that for a constant $C>0$,
\begin{align*}
\Prob(X>tx,Y>ty) = \{C+o(1)\}\exp\{-\kappa(x,y)t\}, \quad  t\to\infty, \quad x,y>0,
\end{align*}
where $\kappa:(0,\infty)^2 \to (0,\infty)$ is a differentiable positive homogeneous function that defines a norm. This special case of the set-up of \citet{WadsworthTawn13} is satisfied by the Morgenstern, inverted extreme value, Ali--Mikhail--Haq, and Pareto copulas, amongst others; see \citet{Heffernan00} for a summary of their extremal dependence properties. All such examples are asymptotically independent, with $\eta=1/\kappa(1,1)$. Let $\kappa_i$ denote the partial derivative of $\kappa$ with respect to its $i$th argument, and similarly let $\kappa_{12}$ denote the joint derivative. Taking $\|(x,y)\|_* = \kappa(x,y)$ gives
\begin{align*}
f_{R,W}(r,w) = \{C+o(1)\}\exp(-r)\left\{\frac{\kappa_1(w,1-w)\kappa_2(w,1-w)}{\kappa(w,1-w)^2}r - \frac{\kappa_{12}(w,1-w)}{\kappa(w,1-w)}\right\}, ~~r\to\infty,
\end{align*}
which satisfies condition~\eqref{eq:denscondn}. Furthermore, since the reciprocal hazard function $h_R(r) = 1 +o(1)$ as $r\to\infty$, $\lambda=0$:  normalized exceedances of $R$ have a  limiting exponential distribution.  The limiting density of $W$ as $r\to\infty$ is 
\begin{align*}
j(w) \propto \frac{\kappa_1(w,1-w)\kappa_2(w,1-w)}{\kappa(w,1-w)^2} ,~~w\in(0,1).
\end{align*}
\end{Example}

\begin{Example}
\label{eg:elliptical}
Let $(X,Y)$ be elliptically distributed, truncated to the positive quadrant, so one can write 
\begin{align*}
(X,Y) = Q \Sigma^{1/2} (U_1,U_2),
\end{align*}
with $\Sigma^{1/2}$ the Cholesky factor of a positive-definite matrix, $(U_1,U_2)$ lying on the part of the unit circle such that $\Sigma^{1/2}(U_1,U_2)$ lies in the positive quadrant, and $Q$ a random variable known as the generator. Then the norm $\|(x,y)\|_*=\{(x,y)\Sigma^{-1}(x,y)^T\}^{1/2}$ returns the variable $Q$, i.e., $R=Q$. Thus we have exact independence of $R$ and $W$, and the density of $W$ is 
\begin{align*}
j(w) \propto \|(w,1-w)\|_*^{-2} = (1-\rho^2)\{w^2-2\rho w(1-w) + (1-w)^2\}^{-1},\quad w\in(0,1).
\end{align*}
The exact form of the limiting distribution for exceedances of $R$ depends on $Q$: \citet{Abdous05} consider extremes of elliptical distributions and provide details on the domain of attraction of the generator. The variables $X$ and $Y$ are asymptotically dependent if and only if $Q$ has regularly varying tails \citep{HultLindskog02}. This links precisely to the asymptotic dependence features described in Section~\ref{sec:ExtremalDependence}. As highlighted by Example~\ref{eg:mvrv}, the norm $\|\cdot\|_*$ may be chosen arbitrarily if $Q$ has a heavy tail, though an advantage of the norm $\|(x,y)\|_*$ is that independence is exact, rather than asymptotic, in the sense of equation~\eqref{eq:RWgen}. The Gaussian is the best-known elliptical distribution; its extremes are asymptotically independent, with $R$ having the Weibull density $f_R(r) = re^{-r^2/2}$ $(r>0)$; thus $\lambda=0$.
\end{Example}

 Like elliptical copulas, Archimedean survival copulas have a radial-angular decomposition, with the pseudo-angles being uniformly distributed on $[0,1]$ \citep{McNeilNeslehova09}. Thus~\eqref{eq:RWgen} is satisfied whenever the radial variable falls into the domain of attraction of a generalized Pareto distribution.

\subsection{Application of~\eqref{eq:RWgen}}
\label{sec:ApplicationOfLimit}
In order to apply~\eqref{eq:RWgen} directly, one must know the (class of) margins $F$, and the (class of) norm $\|\cdot\|_*$, to which it applies. The basis of statistical procedures assuming asymptotic dependence is that any choice of heavy-tailed margins and norm will lead to a limit, and so that choice is arbitrary. If asymptotic dependence cannot be assumed, then the correct class of marginal distributions and the correct norm must be chosen, and this makes direct exploitation of~\eqref{eq:RWgen} challenging. One might choose among marginal classes based on some measure of fit, but this would not account for uncertainty in the dependence class. For this reason we aim to construct a model having the essential features of~\eqref{eq:modapprox}.

\section{Model}
\label{sec:Model}
\subsection{Introduction}
\label{sec:ModelIntro}

We use the observations of Section~\ref{sec:Motivation}, and in particular equation~\eqref{eq:modapprox}, to motivate a model that can capture both asymptotic dependence and asymptotic independence. Consider the dependence structure of
\begin{align}
\label{eq:mod1}
\begin{split}
 &(A,B) = S (V_1, V_2),\\
(V_1,V_2) = (V,1-V)/\|(V,1-V)\|_m \in~& \mathcal{S}^m=\{\bm{v}\in \mathbb{R}^2_+ : \|\bm{v}\|_m = 1\},\quad V\sim F_V~\independent~ S\sim \mbox{GP}(1, \lambda),
\end{split}
 \end{align}
where $F_V$ is a distribution defined on $[0,1]$. The norm $\|\cdot\|_m$ and distribution $F_V$ are modelling choices; $\lambda$ and any parameters of $F_V$ are to be inferred. Model~\eqref{eq:mod1} reflects the structure of~\eqref{eq:modapprox}, which provides an asymptotic representation of the extremes of a wide variety of dependence structures. As we show in Section~\ref{sec:ExtremalDependence}, the dependence structure of \eqref{eq:mod1} is broad enough to capture both types of extremal dependence structures. Although~\eqref{eq:mod1} is motivated by~\eqref{eq:modapprox}, we adopt different notation in order to emphasize that the former is a modelling approach rather than than an assumption on the underlying random vector. 

 Model~\eqref{eq:mod1} has parameters that are common to the margins and dependence structure, but we are interested only in exploiting its copula,
\begin{align}
  C(u_1,u_2) = F_{A,B}\{F_A^{-1}(u_1),F_B^{-1}(u_2)\}, \label{eq:copula}
\end{align}
where, $F_{A,B}$, $F_A$, and $F_B$ are the joint and marginal distribution functions of~\eqref{eq:mod1}. We refer to $F_A,F_B$ as \emph{pseudo-marginals} throughout, as they are unrelated to the true marginals of the observable random vector, reflecting only those in which the factorization~\eqref{eq:mod1} holds best for the extremes. 

Representation~\eqref{eq:modapprox} holds when a suitable pseudo-radial variable is large. By analogy, it is reasonable to assume that~\eqref{eq:mod1} holds only when some norm of the variables is large. This will be implemented in our inference strategy, explained in Section~\ref{sec:Inference}. Thus, if the observed vector $(Z_1,Z_2)$ has joint distribution function $F_{1,2}$, then we suppose for all sufficiently extreme observations that $F_{1,2}(z_1,z_2) \approx C\{F_1(z_1),F_2(z_2)\}$, with $C$ as in~\eqref{eq:copula}. Finally note that the fact that $A$ and  $B$ may have different margins is not incompatible with the spirit of~\eqref{eq:modapprox}, as the margins therein are those of $(X,Y)$ given that $ \|(X,Y)\|_*>0$, which may be unequal if the dependence structure is asymmetric.

\subsection{Extremal dependence properties}
\label{sec:ExtremalDependence}
We detail the extremal dependence properties of the model~\eqref{eq:mod1} under some mild restrictions on the types of norm considered and the support of $V$. Proofs of all propositions may be found in  Appendix~\ref{sec:AuxiliaryResults}. The following conditions on $\|\cdot\|_m$ are imposed throughout this section.

\begin{Cond}[Symmetry]
\label{Cond1}
$\|(x,y)\|_m = \|(y,x)\|_m$.
\end{Cond}

\begin{Cond}[Boundary]
\label{Cond2}
$\|(x,y)\|_m \geq \|(x,y)\|_\infty$.
\end{Cond}

\begin{Cond}[Equality with $L_\infty$]
\label{Cond3}
$\|(x_0,y_0)\|_m=\|(x_0,y_0)\|_\infty$ for some $x_0\neq y_0$.
\end{Cond}

These conditions specify ranges for the marginal projections $V_1=V/\|(V,1-V)\|_m$, and $V_2=(1-V)/\|(V,1-V)\|_m$ to be $[0,1]$. In particular the mapping $T:[0,1]\to[0,1]$ given by $T(v)=v/\|(v,1-v)\|_m$ is surjective. Condition~\ref{Cond3} imposes that if equality with $\|\cdot\|_\infty$ occurs at $(1,1)$, then since we must also have equality somewhere off the diagonal, the norm must behave locally like $\|\cdot\|_\infty$ around $(1,1)$, by convexity. This specifically rules out cases such as $\|(x,y)\|_m = \max\{ax+(1-a)y,ay+(1-a)x\}$, $a>1$, for which $\|(1,1)\|_m=1$, but which does not behave locally like the $L_\infty$ norm; these can induce dependence properties different from those claimed under Condition~\ref{Cond3}.

We focus on the dependence measures $\chi$ (equation~\eqref{eq:chi}) and $\eta$ (equation~\eqref{eq:eta}) and the function $\kappa$ (equation~\eqref{eq:WT}). These were defined following a transformation of the variables to standard Pareto margins, but for exposition of calculation, here we will exploit the equivalence $\Prob(X_P>t^\beta,Y_P>t^\gamma)= \Prob\{A > q_A(t^\beta), B> q_B(t^\gamma)\}$, where $q_{i}(t):=F_i^{-1}(1-1/t)$ ($t\geq 1$, $i\in\{A,B\}$) is the $1-1/t$ quantile function. \citet{WadsworthTawn13} show that under asymptotic dependence, if~\eqref{eq:WT} holds, then $\kappa(\beta,\gamma) \equiv \max(\beta,\gamma)$, whereas more interesting structures are obtained under asymptotic independence. The dependence structure of asymptotically dependent distributions is described by the ray dependence function $d$ or distribution $H$ in equation~\eqref{eq:dH}. We discuss these below, also giving the corresponding quantities for the Ramos--Ledford framework under hidden regular variation.

The marginal and joint survivor functions are key to the study of dependence. The former can be expressed as $\Prob(A>x) = \E\{\Prob(S V_1 > x\mid V_1)\}$ and $\Prob(B>y) = \E\{\Prob(S V_2 > y\mid V_2)\}$, where, noting the link between $(V_1,V_2)$ and $V$, $\E$ denotes expectation with respect to $V$. This provides
 \begin{align}
 \label{eq:marg}
 \Prob(A>x) &=  \E\left\{ (1+ \lambda x / V_1)_+^{-1/\lambda} \right\}, & \Prob(B>y) &=  \E\left\{ (1+ \lambda y / V_2)_+^{-1/\lambda} \right\}.
 \end{align}
 The joint survivor function can likewise be expressed as
 \begin{align}
 \label{eq:joint}
 \Prob(A>x,B>y) &=\E \left[\{1+\lambda \max(x / V_1,y/V_2)\}_+^{-1/\lambda} \right].
 \end{align}

Below we present $\chi$, $\eta$ and $\kappa(\beta,\gamma)$ for the different ranges of $\lambda$, and types of norm under consideration. For all cases we assume:

\begin{As}
 \label{As_supp}
 The distribution function of $V$, $F_V:[0,1] \to [0,1]$, is continuous and strictly increasing. 

\end{As}
Equivalently the measure associated to $F_V$ has no point masses and its support is the entire unit interval. With $T$ as defined above, define $\lep :=\inf\{v\in[0,1]:T(v)=1\}$ and $\rep :=\sup\{v\in[0,1]:T(v)=1\}$.

\begin{Case}[$\lambda>0$] \label{C:lplf}

Define the positive quantity
\begin{align}
 \chi_{\lambda} = \E\left[\min\left\{V_1^{1/\lambda} / \E(V_1^{1/\lambda}), V_2^{1/\lambda} / \E(V_2^{1/\lambda})\right\}\right]. \label{eq:chi_lambda}
\end{align}

\begin{prop}
\label{l>0xyj:prop}
If $\beta,\gamma>0$, then 
\[
{\Prob}\bigl\{A > \xq,\,B > \yq\bigr\}=t^{-\max(\beta,\gamma)}\xyj(t),\quad t\geq 1, 
\]
where $\xyj$ is slowly varying at infinity. Furthermore, $\xyj(t)\to\chi_\lambda$ as $t\to\infty$ if $\beta=\gamma$, and $\xyj(t)\to1$ otherwise.
\end{prop}

It is an immediate corollary that $\eta = 1$, and $\chi=\chi_\lambda>0$. However, for any fixed $F_V$, as $\lambda \to 0$ the dependence weakens to asymptotic independence, by the following: 
\begin{prop}
\label{chi:prop}
Given a fixed $F_V$, $\chi_\lambda\to0$ as $\lambda\to0^+$.
\end{prop}

\end{Case}

\begin{Remark}
\label{rmk:1}
 The ray dependence function~\eqref{eq:dH} for $\lambda>0$ is
\begin{align*}
 d(q) = \frac{1}{\chi_\lambda} \E\left[\min\left\{\frac{V_1^{1/\lambda}}{\E(V_1^{1/\lambda})}\left(\frac{1-q}{q}\right)^{1/2},\frac{V_2^{1/\lambda}}{\E(V_2^{1/\lambda})}\left(\frac{q}{1-q}\right)^{1/2}\right\}\right], \quad q\in[0,1].
 \end{align*}
 If $F_V$ has a Lebesgue density $f_V$, the associated spectral density $h(w) = \dsp H(w)/\dsp w$ is given by
 \begin{align*}
 h(w;\lambda,f_V) = \frac{1}{2} \frac{\lambda^{1-1/\lambda} w^{\lambda-1}(1-w)^{\lambda-1}\mu_1^\lambda\mu_2^\lambda}{\|(w\mu_1)^\lambda,((1-w)\mu_2)^\lambda\|_m^{1/\lambda}\{(w\mu_1)^\lambda+((1-w)\mu_2)^\lambda\}^2} \times f_V\left\{\frac{(\mu_1 w)^\lambda}{(w\mu_1)^\lambda+((1-w)\mu_2)^\lambda}\right\},
\end{align*}
with $\mu_1 = \E(V_1^{1/\lambda})/\lambda^{1/\lambda}$, $\mu_2 = \E(V_2^{1/\lambda})/\lambda^{1/\lambda}$. This satisfies $\int_0^1 w h(w;\lambda,f_V) dw= 1/2$, a necessary moment constraint on $H$, even if $\int_0^1 v f_V(v) \dsp v \neq 1/2$. Justification for these forms is given in Appendix~\ref{sec:AppB}.

\end{Remark}

\begin{Case}[$\lambda=0$]\label{C:l0}

\begin{prop}
\label{xyj:prop}
Let $\beta,\gamma>0$, and define $\omega:=\beta/(\beta+\gamma)$. Then 
\[
\Prob\bigl\{A>\xq,\,B>\yq\bigr\}=t^{-\kappa(\beta,\gamma)}\xyj(t),\quad t\geq 1, 
\]
where $\xyj$ is slowly varying at infinity, and 
\[
 \kappa(\beta,\gamma) = \left\{\begin{array}{ll} \norm{(\beta,\gamma)}, & \omega \in[1-\lep,\lep]\\
                                        \|(\beta,\gamma)\|_{\infty}, &\mbox{otherwise.}
                                      \end{array}\right.
\]

\end{prop}

It is an immediate corollary that $\eta = \|(1,1)\|_m^{-1}$. When $\eta<1$ then $\chi=0$, i.e., we have asymptotic independence. When $\eta=1$, then $\chi=\lim_{t\to\infty}\xyj(t)$ when $\beta = \gamma$. Proposition~\ref{prop:lam0chi} in Appendix~\ref{sec:AuxiliaryResults} states that this limit is still zero, i.e., we still have asymptotic independence.

\end{Case}

\begin{Case}[$\lambda<0$ and $\|(1,1)\|_m= \|(1,1)\|_\infty$]
\label{C:lnli}
For this case only, we further assume:
\begin{As}
 \label{As_cd}
 $F_V$ is continuously differentiable near $1/2$ with $F_V'(1/2)>0$. 
\end{As}

\begin{prop}
\label{l<0xyj:prop}
If $\beta,\gamma>0$, then 
\[
\Prob\bigl\{A>\xq,\,B> \yq\bigr\}=t^{-\kappa(\beta,\gamma)}\xyj(t),\quad t\geq 1, 
\]
where $\kappa(\beta,\gamma)=(1+\lambda)\max(\beta,\gamma)-\lambda (\beta+\gamma)$ and $\xyj$ is slowly varying at infinity with
\begin{equation}
\label{xyjlim:eq}
\lim_{t\to\infty}\xyj(t) = \frac{F_V'(1/2)}{4} \times
\begin{cases}
\upm^\lambda\lowm^{-1},&\beta<\gamma,\\
\bigl\{\min(\upm,\lowm)^\lambda-\frac{1+\lambda}{1-\lambda}\max(\upm,\lowm)^{\lambda}\bigr\}\max(\upm,\lowm)^{-1}, &\beta=\gamma, \\
\lowm^\lambda\upm^{-1},&\beta>\gamma,
\end{cases}
\end{equation}
for $\upm = \Prob(V\in[\lep,\rep])$ and $\lowm = \Prob(V\in[1-\rep,1-\lep])$.
\end{prop}

A corollary when $\beta=\gamma$ is that $\eta=(1-\lambda)^{-1}$. Since $\eta<1$ we must have $\chi=0$, asymptotic independence.
\end{Case}

\begin{Remark}
\label{rmk:2}
 The ray dependence function~\eqref{eq:dHeta} in this case is
 \begin{align*}
  d(q) = \{q(1-q)\}^{\frac{1-\lambda}{2}} \frac{\min\{q\upm,(1-q)\lowm\}^{\lambda}\max\{q\upm,(1-q)\lowm\}^{-1} -\frac{1+\lambda}{1-\lambda}\max\{q\upm,(1-q)\lowm\}^{\lambda-1}}{\min(\upm,\lowm)^{\lambda}\max(\upm,\lowm)^{-1} -\frac{1+\lambda}{1-\lambda}\max(\upm,\lowm)^{\lambda-1}}; 
 \end{align*}
see Appendix~\ref{sec:AppB}. The density of the associated measure $H_\eta$ can be calculated as in \citet[][Section 9.5.3]{Beirlant04}.
\end{Remark}

\begin{Case}[$\lambda<0$ and $\|(1,1)\|_m > \|(1,1)\|_\infty$]
\label{C:lnlf}
In this case $\chi=0$, but the regular variation assumptions~\eqref{eq:eta} and~\eqref{eq:WT} are not satisfied. The marginal densities have upper endpoint $-1/\lambda$, i.e., $\xq, \yq \to -1/\lambda$ as $t\to\infty$, but the upper endpoint of the joint survivor function is strictly below $-1/\lambda$, as can be seen by substituting $x=\xq$, $y=\yq$ in~\eqref{eq:joint}; this probability will be exactly zero whenever
\begin{align}
 \max\left\{\xq/V_1, \yq/V_2\right\} \geq -1/\lambda,\quad (V_1,V_2) \in \mathcal{S}^m. \label{eq:lnineq1}
\end{align}
For $a,b,c,d>0$, $\max(a/b,c/d) \leq \max(a,c)/\min(b,d)$, yielding $\max(a,c) \geq \min(b,d)\max(a/b,c/d)$, so 
\[\max\{\xq/V_1,\yq/V_2\} \geq \min\{\xq,\yq\}\max(1/V_1,1/V_2).\]
Moreover $\max(1/V_1,1/V_2) = 1/\min(V_1,V_2) \geq \|(1,1)\|_m$, since  $\min(V_1,V_2)$ is largest when $V=1/2$. Combining these two observations we have
\[
 \max\left\{\xq/V_1, \yq/V_2\right\} \geq \min\left\{\xq,\yq\right\}\|(1,1)\|_m \to - \|(1,1)\|_m /\lambda > -1/\lambda,\quad t\to\infty, 
\]
so there is a $t_0<\infty$ such that~\eqref{eq:lnineq1} is satisfied for all $t>t_0$. It follows that $\chi=0$, whereas $\eta$ and $\kappa(\beta,\gamma)$ are ill-defined.

\end{Case}

Propositions~\ref{l>0xyj:prop},~\ref{xyj:prop},~\ref{l<0xyj:prop} and Remark~\ref{rmk:1} show how different combinations of $\lambda$, $F_V$ and $ \|\cdot\|_m$ influence extremal dependence properties, under the assumed conditions on the support of $V$ and type of norm. To summarize: asymptotic dependence is present when $\lambda>0$, with the dependence then described by $d(q)$ given in Remark~\ref{rmk:1}, determined by $\lambda$, $F_V$ and $\|\cdot\|_m$. Asymptotic independence is present when $\lambda\leq 0$; for $\lambda=0$, $\kappa$ is determined by the shape of $\|\cdot\|_m$, while for $\lambda<0$, hidden regular variation only arises if $\|(1,1)\|_m=1$. Overlap in dependence structures might seem to arise when $\lambda=0$ and $\|(\beta,\gamma)\|_m = \delta(\beta+\gamma) + (1-\delta)\max(\beta,\gamma)$, $\delta \in (0,1]$, since this matches the case $\lambda \in[-1,0)$ and $\|(1,1)\|_m = \|(1,1)\|_\infty$. However, Proposition~\ref{l<0xyj:prop} shows that in general the slowly varying function arising when $\lambda<0$ depends on the properties of the norm $\|\cdot\|_m$ used for a fixed distribution $F_V$, whereas the slowly varying function arising when $\lambda=0$  cannot change in this way.

 \subsection{Transition between dependence classes}
\label{sec:Transition}

Due to the focus on limits such as~\eqref{eq:chi}, the classification between asymptotic dependence and asymptotic independence is viewed as dichotomous: either the joint and marginal survivor probabilities decay at the same rate or they do not. Where existing modelling approaches are suitable for both dependence types, the transition between them occurs on the boundary of a parameter space, inducing an undesirable discontinuity in the extremal dependence features. For example, consider $\chi(u):=\Prob\{F_1(Z_1)>u \mid F_2(Z_2)>u\}$ ($u\in[0,1]$). In the Ramos--Ledford--Tawn approach, when $\eta=1$ there is an instant ``jump'' to $\chi(u)\equiv \chi>0$ for all $u$ above the level at which the model is assumed to hold, whereas when $\eta<1$,  $\chi(u)\to0$ as $u\to 1$. Similarly in the Heffernan--Tawn model, when $\alpha=0,\beta=1$, the value of $\chi(u)\equiv 1- \int_{0}^\infty G(-v)e^{-v}\,\mbox{d}v$ for all $u$ above the level at which the model is assumed to hold, where $G$ is as in limit~\eqref{eq:HT}, whereas $\chi(u)\to 0$ for all other values of $(\alpha,\beta)$. Consequently, \emph{any} decrease in an empirically estimated $\chi(u)$ suggests that asymptotic independence will be inferred under the Ramos--Ledford--Tawn and Heffernan--Tawn models.

  An elegant feature of model~\eqref{eq:mod1} is the smoothness of the transitions across dependence classes in $\lambda$, and the fact that asymptotic independence or dependence does not occur at boundary points for $\lambda$. In particular when $\lambda \to 0^+$, the function $\chi_\lambda$ defined in~\eqref{eq:chi_lambda} tends to zero, and the value of the function $\chi(u)\equiv\chi_{\lambda}(u)$ discussed in Section~\ref{sec:Transition} may depend on $u$ in regions where the model holds, thereby smoothing out some of the discontinuity discussed above. Furthermore, if $\|(1,1)\|_m = 1$, achieved if we set $\|\cdot\|_m=\|\cdot\|_\infty$, then $\chi_{\lambda}\to 0$ as $\lambda\to 0^+$ and $\eta$ decreases from 1 at $\lambda=0$ towards $0$ as $\lambda\to -\infty$. In this sense the model smoothly interpolates across the dependence classes. We will adopt these modelling choices in Section~\ref{sec:Inference}.

\section{Inference}
\label{sec:Inference}
\subsection{Likelihood and parameterization}
\label{sec:Likelihood}
 We now consider fitting~\eqref{eq:mod1} as a dependence model for extreme bivariate data by likelihood methods. Let $F_A, F_B$ and $f_A, f_B>0$ denote the pseudo-marginal distribution and density functions respectively, and let $f_{A,B}$ denote the joint density of $(A,B)$. The density corresponding to the copula $C(u_1,u_2)$ is
 \begin{align*}
 c(u_1,u_2) = \frac{f_{A,B}\{F_A^{-1}(u_1),F_B^{-1}(u_2)\}}{f_A\{F_A^{-1}(u_1)\} f_B\{F_B^{-1}(u_2)\} }, \quad 0\leq u_1,u_2\leq 1.
\end{align*}
Recall that $(V_1,V_2) = (V,1-V)/\|(V,1-V)\|_m $; we assume that $V$ has a Lebesgue density (thus Assumptions~\ref{As_supp} and~\ref{As_cd} are satisfied), denoted by $f_V$. Using the independence of $S$ and $V$ we obtain the joint density 
\begin{align*}
f_{A,B}(x,y) = \frac{\|(x,y)\|_m}{(x+y)^2} \left\{1 + \lambda \|(x,y)\|_m\right\}_+^{-1/\lambda - 1} f_V\left(\frac{x}{x+y}\right), \quad x,y>0.
\end{align*}
The pseudo-marginal density and distribution functions required to compute $c(u_1,u_2)$ are not explicit, requiring numerical evaluation of a one-dimensional integral.

We only wish to use model~\eqref{eq:mod1} for extreme dependence, so we must censor non-extreme data. Since the margins and dependence have a common parameterization, it is only straightforward to censor on regions that remain of the same form under marginal transformation. We therefore choose to censor data for which the maximum value on the uniform marginal scale is less than some $u$ close to 1. This translates to the uncensored variables having $\max(A,B)$ large, and by equivalence of norms, any $\|(A,B)\|_m$ will also be large. Thus the likelihood that we use for independent pairs $(u_{1,1},u_{2,1}),\ldots, (u_{1,n},u_{2,n})$ with uniform margins is
\begin{align} 
L(\bm{\zeta}) = \prod_{i: \max(u_{1,i},u_{2,i})>u} c(u_{1,i},u_{2,i};\bm{\zeta}) \prod_{i: \max(u_{1,i},u_{2,i})\leq u} C(u,u;\bm{\zeta}), \label{eq:lhood}
\end{align}
with $\bm{\zeta}$ a parameter vector. In practice the data must be transformed to uniform margins using the probability integral transform. One possibility is semiparametric transformation, using the empirical distribution below a high threshold and the asymptotically-motivated generalized Pareto distribution above it \citep{ColesTawn91}. A simpler alternative is to use the empirical distribution function throughout. The properties of censored two-stage parametric and semiparametric maximum likelihood estimators of copula parameters are explored in \citet{ShihLouis95}.

In this implementation, we constrain $\lambda\leq 1$. In order to fit the model, points must be transformed onto $A$, $B$ pseudo-margins using numerical inversion; if $\lambda$ is large,  then numerical instabilities may arise because the pseudo-margins are heavy-tailed. Considering the form of $h(\cdot;\lambda,f_V)$ given in Remark~\ref{rmk:1}, this still yields a slightly richer class of spectral densities than those defined simply by $f_V$. The complete set of parameters is determined by the choice of $f_V$ and any parameterization of the norm $\|\cdot\|_m$. Below we take
\begin{align}
V\sim \mbox{Beta}(\alpha,\alpha)\quad \mbox{and}\quad \|\cdot\|_m = \|\cdot\|_\infty, \label{eq:modchoice}
\end{align}
giving $\bm{\zeta}=(\lambda,\alpha)$. The beta distribution is chosen for its simplicity and flexibility of shape, but might be replaced by other distributions. As mentioned in Section~\ref{sec:ExtremalDependence}, \eqref{eq:modchoice} permits all possible $\chi$ and $\eta$ values; it also provides a simple model for the dependence structure in both asymptotic independence and dependence frameworks, through the attainable forms of $\kappa$, and ray dependence function $d$. Although~\eqref{eq:modchoice} represents a misspecification for each of the dependence structures to be used in Section~\ref{sec:Simulation}, our numerical results  suggest that it works reasonably well.

Recalling Section~\ref{sec:Uniqueness}, the choice of a fixed norm in model~\eqref{eq:mod1} is not as restrictive as might first appear. Since the extremal dependence depends on the combination of $\lambda$, $\|\cdot\|_m$ and the distribution of $V$, the fixing of the norm can be offset by the other model elements to yield a good representation of the data anyway.

An R package for fitting and checking model~\eqref{eq:mod1}, \texttt{EVcopula}, is available at \texttt{www.lancaster.ac.uk/$\sim$wadswojl/}.

\subsection{Parameter Identifiability}
\label{sec:Identifiability}

The parameters $(\lambda,\alpha)$ of the model defined by~\eqref{eq:mod1} and~\eqref{eq:modchoice} exhibit negative association, as increasing either parameter whilst fixing the other gives stronger dependence. When the data derive from an asymptotically dependent random vector exhibiting multivariate regular variation, this trade-off may be particularly strong, because  each $\lambda>0$ leads to a spectral density (in the sense described in Section 1, derived using standard Pareto margins and the $L_1$ norm) $h(\cdot;\lambda,f_V)$, as detailed in Remark~\ref{rmk:1}. With the modelling choices in~\eqref{eq:modchoice} the spectral density $h(\cdot;\lambda,f_V)$ simplifies to
\begin{align*}
 h(w;\lambda,\alpha)&=\frac{1}{2\mu} \frac{\lambda^{1-1/\lambda} w^{\lambda\alpha-1}(1-w)^{\lambda\alpha-1}}{\{w^\lambda+(1-w)^\lambda\}^{2\alpha} \max(w,1-w)}\frac{\Gamma(2\alpha)}{\Gamma(\alpha)^2}, \quad 0<w<1, 
\end{align*}
with $\mu=\mu_1=\mu_2$ due to symmetry. A dominant factor in maximum likelihood estimation of $(\lambda,\alpha)$ is thus the combination of these parameters giving a spectral density most similar to the underlying truth. Although the parameters have different roles, in practice there are many combinations that yield a similar $h(\cdot;\lambda,\alpha)$. To determine if the resulting identifiability issues matter in applications, we suggest inspection of the joint log-likelihood surface for $(\lambda,\alpha)$; we implement this in Section~\ref{sec:Applications}. More generally for other norms and choices of $f_V$, parameter identifiability must be considered.

\subsection{Simulation}
\label{sec:Simulation}

For three different dependence structures, we estimate the probability of lying in rectangular-shaped sets $(u_1,v_1)\times(u_2,v_2)$ on the copula scale, where $u_1<v_1$, $u_2<v_2$, and $u_2$ represents an extreme quantile. We call them Set~1, \ldots, Set~5, with $(u_1,v_1) = (0.05,0.2), (0.2,0.4), (0.4,0.6), (0.6, 0.8), (0.8,0.9999)$, and $(u_2,v_2)=(0.995,0.99995)$, respectively. We compare our method to that of \citet{HeffernanTawn04}, the only other approach easily able to estimate probabilities when the components may not both be extreme.

We simulate 100 replicate samples of size 1000 from (i) the bivariate extreme value distribution with symmetric logistic dependence structure \citep{ColesTawn91}; (ii) the inverted copula of (i) \citep{LedfordTawn97}; and (iii) the bivariate normal distribution. The first is asymptotically dependent, and the others are asymptotically independent. We use dependence parameters $\{0.2,0.4,0.6,0.8\}$, representing decreasing dependence for (i) and (ii) and increasing dependence for (iii): we label the dependence levels from 1--4 in order of increasing strength. The censoring threshold in likelihood~\eqref{eq:lhood} was $u=0.95$, and the data were transformed to uniformity using the empirical distribution function. Estimation for the \citet{HeffernanTawn04} method was based on all data whose $Y$ coordinate exceeded a 90\% quantile threshold. 

Table~\ref{tab:allresults} displays the root mean squared errors (RMSEs) of the log of all non-zero estimated probabilities. For our model, we define a probability to be zero if its estimate is less than twice machine epsilon in R, since numerical procedures are involved in the calculations; this can occasionally produce negative numbers, which we also set to zero. The number of probabilities estimated as zero is also provided in the table.  Overall the new model produces estimates with lower RMSEs than the Heffernan--Tawn model. Any exceptions arise when the Heffernan--Tawn model estimates only a very few non-zero probabilities. In general, estimation for sets closer to one of the axes is better when the dependence is lower. This seems natural as when dependence is high, few if any points in a dataset will be observed near the axes. Both models perform poorly under strong dependence for sets near the axes. Future work could explore whether a more sophisticated implementation of our approach, such as allowing different $f_V$, $\|\cdot\|_m$, or changing the censoring scheme, improves this. 

\begin{table}[t]
\centering

\caption{RMSEs of non-zero log-probabilities and number of zero estimated probabilities for the new model (New) and Heffernan and Tawn model (HT) for dependence structures (i)--(iii).}
\begin{tabular}{|l|l|lllll|lllll|}\hline
&&\multicolumn{5}{|c|}{RMSE}&\multicolumn{5}{c|}{Number of zeroes}\\\hline
Dep. / Method & Level & Set 1 & Set 2 & Set 3 & Set 4 & Set 5 & Set 1 & Set 2 & Set 3 & Set 4 & Set 5\\ \hline
(i) / New & 1 & 0.47 & 0.39 & 0.33 & 0.28 & 0.095 & 1 & 0 & 0 & 0 & 0 \\ 
&  2 & 1.70 & 1.00 & 0.71 & 0.52 & 0.023 & 0 & 0 & 0 & 0 & 0 \\ 
&  3 & 5.30 & 4.30 & 3.30 & 1.90 & 0.0011 & 41 & 2 & 0 & 0 & 0 \\ 
&  4 & 13.00 & 11.00 & 6.90 & 8.90 & 0.0009 & 95 & 97 & 85 & 62 & 0 \\ 
\hline
(i) / HT &  1 & 1.70 & 1.40 & 1.40 & 0.69 & 0.17 & 45 & 19 & 8 & 0 & 0 \\ 
&  2 & 1.10 & 1.10 & 1.00 & 1.50 & 0.033 & 98 & 87 & 57 & 18 & 0 \\ 
&  3 & -- & 4.40 & 3.70 & 2.00 & 0.02 & 100 & 99 & 99 & 89 & 0 \\ 
&  4 & -- & -- & -- & -- & 0.018 & 100 & 100 & 100 & 100 & 0 \\ 
\hline
(ii) / New&  1 & 0.25 & 0.15 & 0.13 & 0.10 & 0.16 & 0 & 0 & 0 & 0 & 0 \\ 
&  2 & 0.53 & 0.27 & 0.24 & 0.19 & 0.11 & 0 & 0 & 0 & 0 & 0 \\ 
&  3 & 2.50 & 1.30 & 0.66 & 0.41 & 0.043 & 5 & 0 & 0 & 0 & 0 \\ 
&  4 & 6.90 & 5.20 & 3.60 & 1.90 & 0.0041 & 20 & 11 & 7 & 0 & 0 \\ 
 \hline
 (ii) / HT & 1 & 1.60 & 0.93 & 0.40 & 0.29 & 0.31 & 16 & 4 & 0 & 0 & 0 \\ 
&   2 & 0.90 & 1.30 & 1.30 & 0.55 & 0.20 & 56 & 20 & 1 & 0 & 0 \\ 
&  3 & 2.10 & 0.92 & 1.20 & 1.20 & 0.067 & 94 & 73 & 30 & 1 & 0 \\ 
&  4 & -- & -- & -- & 1.20 & 0.02 & 100 & 100 & 100 & 82 & 0 \\  \hline 

(iii) / New & 1 & 0.24 & 0.17 & 0.14 & 0.13 & 0.20 & 0& 0& 0& 0& 0\\ 
&  2 & 0.52 & 0.38 & 0.29 & 0.18 & 0.17 & 0& 0& 0& 0& 0\\ 
&  3 & 1.20 & 0.75 & 0.56 & 0.35 & 0.095 & 1 & 0& 0& 0& 0\\ 
&  4 & 3.60 & 2.10 & 1.40 & 0.88 & 0.016 & 19 & 3 & 0& 0& 0\\ \hline
(iii) / HT & 1 & 1.70 & 0.86 & 0.41 & 0.29 & 0.37 & 13 & 2 & 0& 0& 0\\ 
&  2 & 1.20 & 1.40 & 0.79 & 0.33 & 0.21 & 51 & 15 & 1 & 0 & 0 \\ 
&  3 & 2.70 & 1.30 & 1.40 & 1.00 & 0.10 & 93 & 69 & 27 & 1 & 0 \\ 
&  4 & -- & -- & 2.30 & 1.50 & 0.021 & 100 & 100 & 96 & 54 & 0 \\
 \hline
\end{tabular}
\label{tab:allresults}
\end{table}

As a diagnostic for the model fit, we also consider the extremal dependence functions $\chi(u)$, defined in Section~\ref{sec:Transition}, and $\bar{\chi}(u):=2\log (1-u)/\log\{\Prob(F_1(Z_1)>u,F_2(Z_2)>u)\} - 1$ for $u\in(0.9,0.999)$ \citep{Coles99}. As $u\to 1$, $\chi(u) \to \chi$, as in~\eqref{eq:chi}, whilst $\bar{\chi}(u) \to \bar{\chi}=2\eta-1 \in [-1,1]$. The value of $\chi$ thus gives some discrimination between different asymptotically dependent copulas, whilst $\bar{\chi}$ can discriminate between different asymptotically independent copulas. As functions of $u$, $\chi(u)$ and $\bar{\chi}(u)$ are useful for checking model fits under either dependence scenario. Figure~\ref{fig:ChiAll} displays pointwise medians and 90\% confidence intervals of $\chi(u), \bar{\chi}(u)$ for each dependence structure and for both methods of inference. Small biases of the new model are typically offset by lower variability and better performance away from the diagonal, i.e., away from the region on which $\chi(u)$ and $\bar{\chi}(u)$ focus.

\begin{figure}[htbp]
\centering
\includegraphics[width=0.95\textwidth, angle=90]{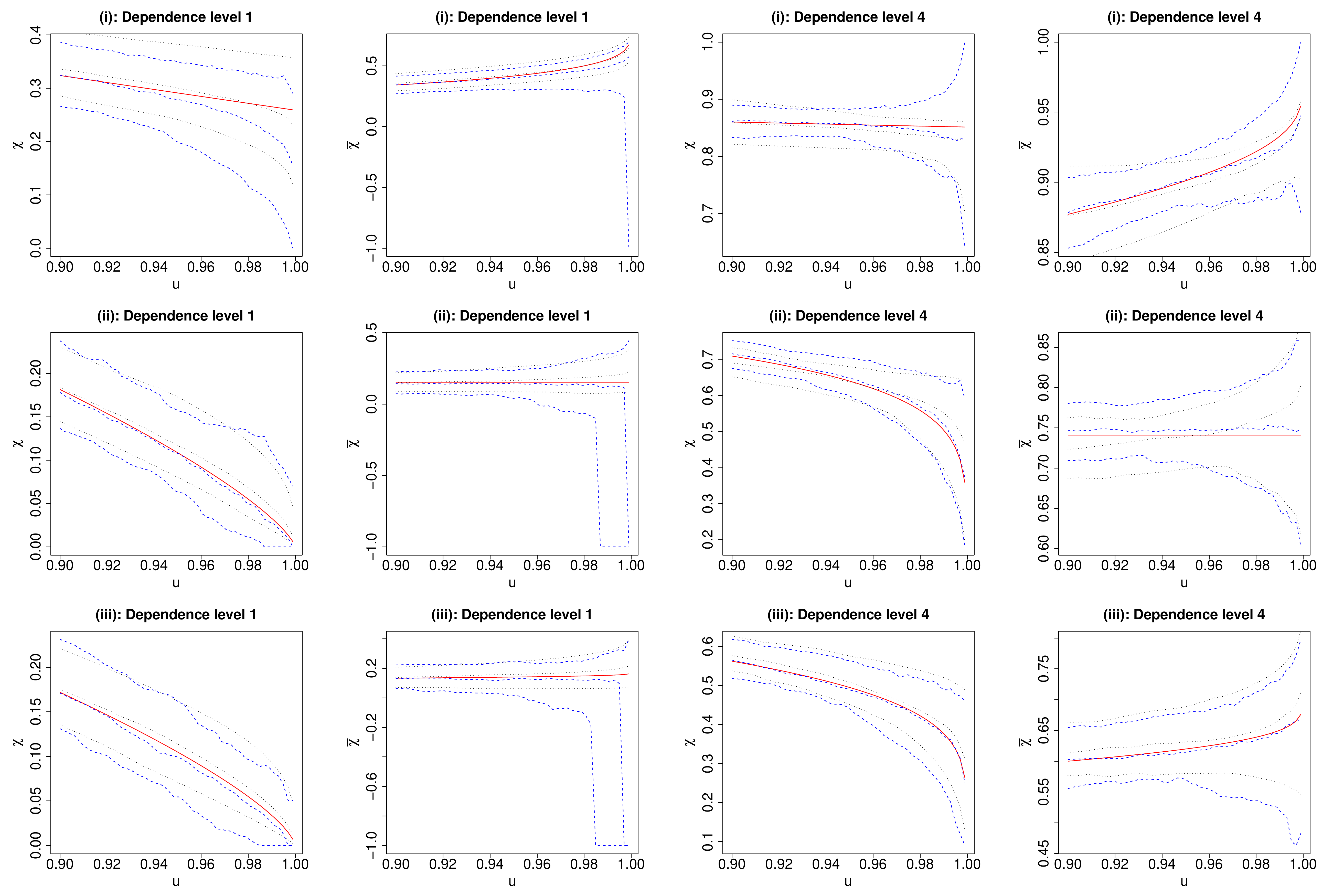}
\caption{Estimates of $\chi(u)$ (left) and $\bar{\chi}(u)$ (right) for dependence levels 1 and 4 of dependence structures (i)--(iii) using the new model (dotted lines) and the Heffernan--Tawn model (dashed lines). The three lines represent pointwise means and upper 95\% and lower 5\% quantiles of the 100 repetitions. Red solid line: true value for the copula. The dependence structures and levels are given as the figure title.}
\label{fig:ChiAll}
\end{figure}

\section{Environmental application}
\label{sec:Applications}
We consider an oceanographic dataset comprising measurements of wave height, surge and wave period recorded at Newlyn, U.K., filtered to correspond to a 15-hour time window for approximate temporal independence, and previously analyzed by \citet{ColesTawn94}, \citet{Bortot00} and \citet{ColesPauli02}. \citet{ColesTawn94} noted the presence of seasonality, which was not taken into account in their, or subsequent, analyses; for ease of comparison we also ignore it. \citet{ColesTawn94} used an asymptotically dependent model for these data, whilst \citet{Bortot00} used an asymptotically independent Gaussian tail model. \citet{ColesPauli02} employed a mixture-type model, able to encompass both dependence types, with asymptotic dependence arising at a boundary point. The literature appears to have reached a consensus that there is strong, but not overwhelming, evidence for asymptotic dependence between wave height and surge, and fairly strong evidence for asymptotic independence between the other two pairs.

Here we fit the simple symmetric model~\eqref{eq:modchoice}, with dependence threshold $u=0.95$ in likelihood~\eqref{eq:lhood}. Marginal transformations to uniformity were carried out using the semiparametric procedure of \citet{ColesTawn91} described in Section~\ref{sec:Likelihood}, but the dependence parameter estimates were almost the same using the fully empirical marginal transformation.

Table~\ref{tab:wave} gives maximum likelihood estimates and confidence intervals for the dependence parameters. The estimate $\hat{\lambda}=0.54$ suggests asymptotic dependence between wave height and surge, whilst the values  $\hat{\lambda}=-0.21$  and $\hat{\lambda}=-0.43$ indicate asymptotic independence for the pairs involving period; this is supported by the confidence intervals. The likelihood surfaces plotted in Figure~\ref{fig:lhood} show that the parameters are identifiable and give an appreciation of the joint asymptotic confidence regions. Figure~\ref{fig:wavechi} shows the empirical and fitted functions $\chi(u)$ and $\bar{\chi}(u)$, which suggest a reasonable fit to the data. Fits from the Heffernan--Tawn model are also displayed, conditioned on each variable in turn, and show potential discrepancies in the inferred strength of the dependence; by having only a single model, we can avoid such discrepancies and the need to decide which variable should be chosen for conditioning upon.

\begin{table}[t]
\begin{center}
\caption{Maximum likelihood estimates and  95\% profile likelihood confidence intervals for $(\lambda,\alpha)$, for the pairs of oceanographic variables.}
\begin{tabular}{lccc}
  \hline
 &  Height--Surge & Height--Period &  Surge--Period  \\ 
  \hline
$\hat{\lambda}$ & $0.54$ $(0.26,0.90)$ & $-0.21$ $(-0.40,-0.03)$ & $-0.43$ $(-0.66,-0.16)$ \\ 
 $\hat{\alpha}$& $0.58$ $(0.40,0.81)$ & $2.21$ $(1.58,3.06)$& $0.68$ $(0.50,0.90)$ \\ 
   \hline
\end{tabular}
   \label{tab:wave}
\end{center}
\end{table}

\begin{figure}[p]
\centering
\includegraphics[width=0.3\textwidth]{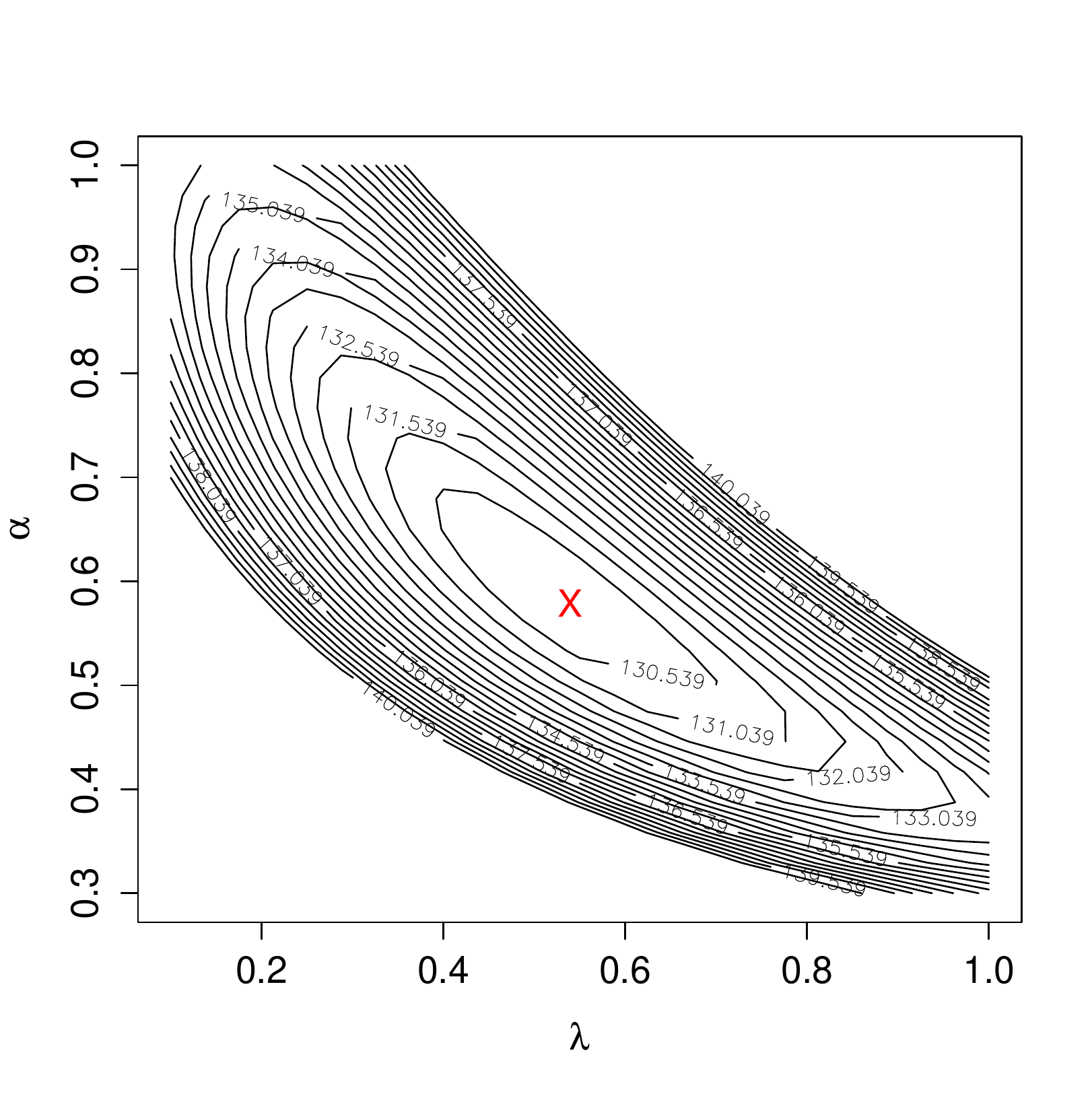}
\includegraphics[width=0.3\textwidth]{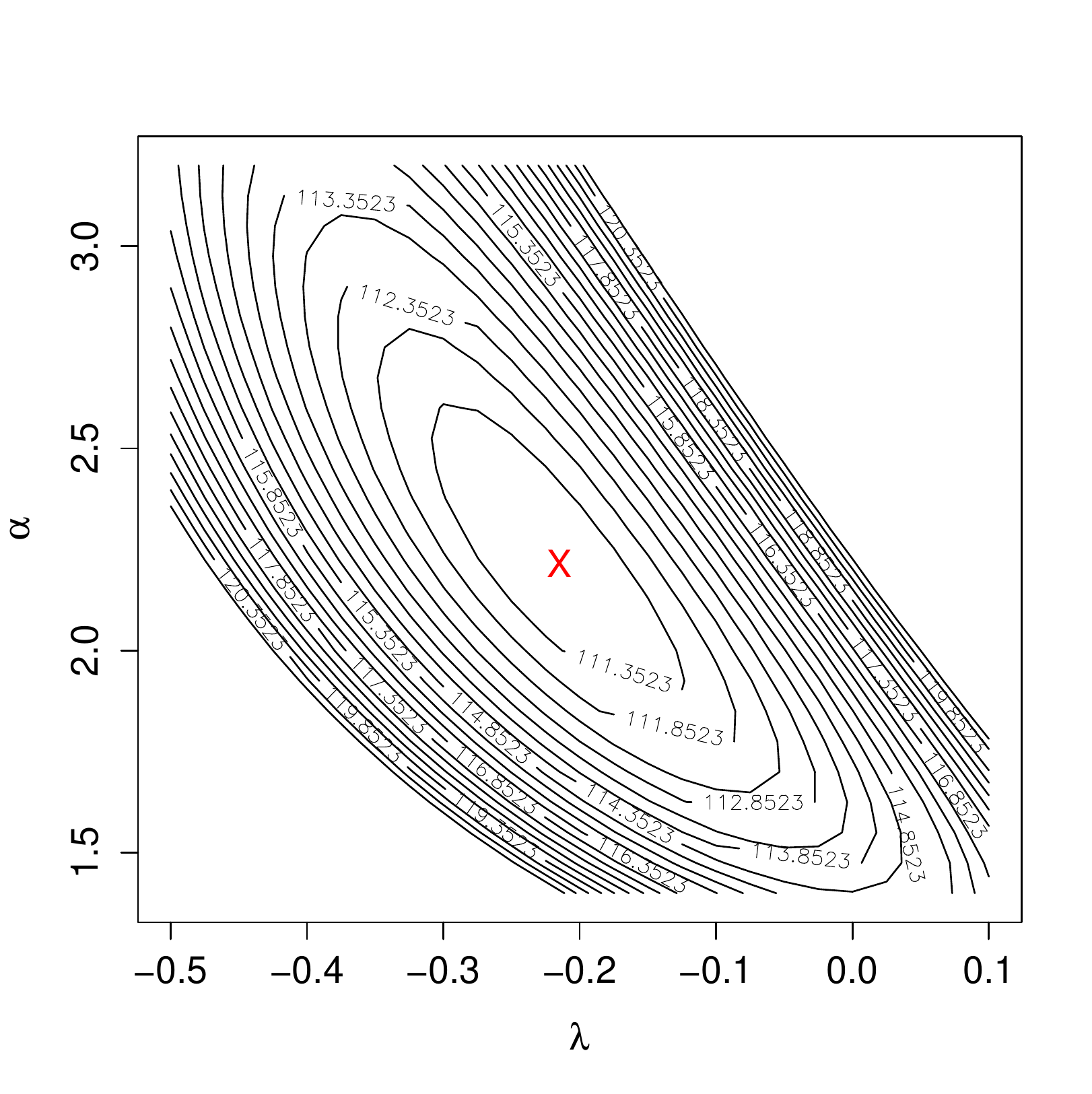}
\includegraphics[width=0.3\textwidth]{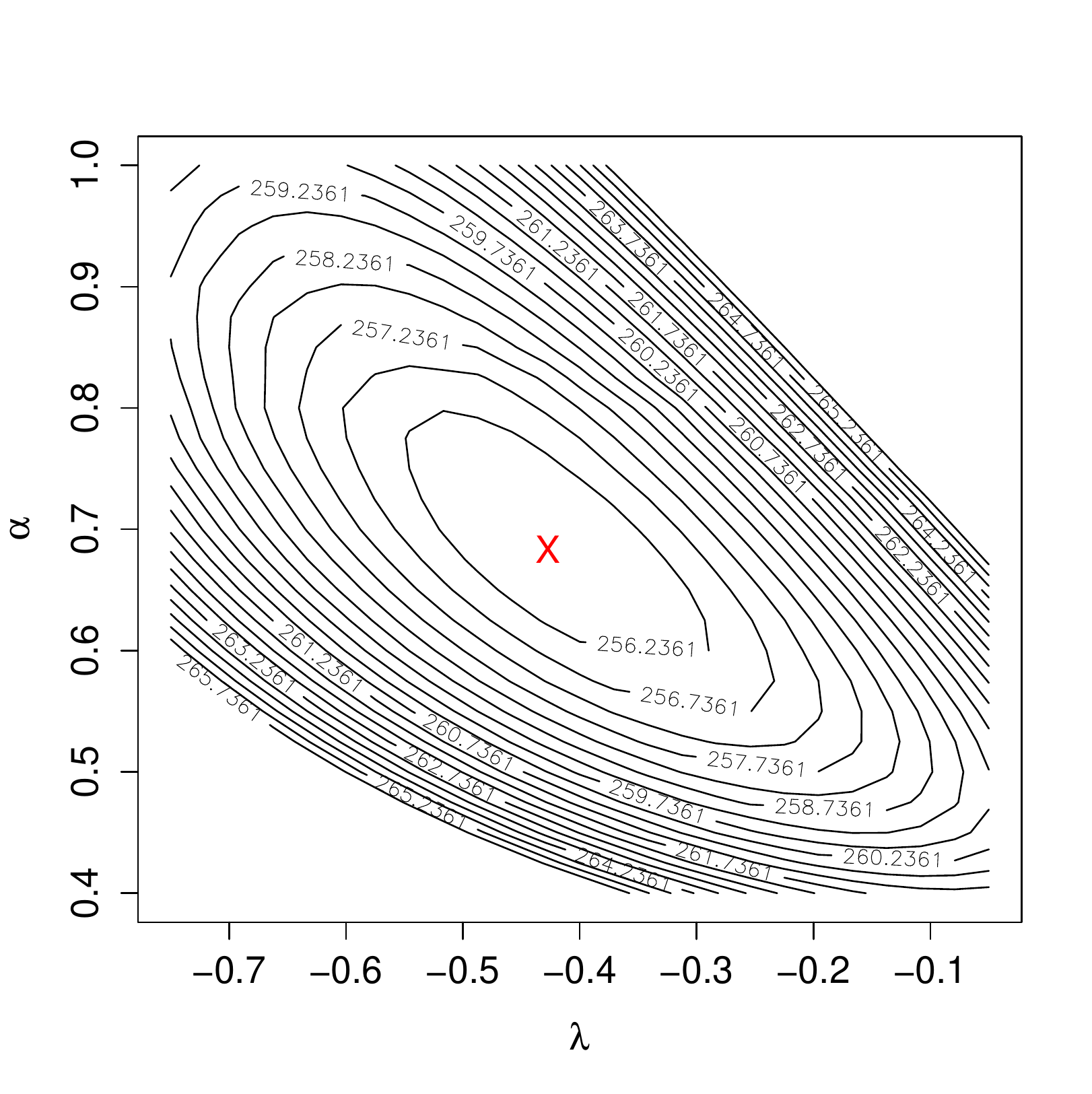}

\caption{Negative log-likelihood surfaces for height--surge, height--period and surge--period, respectively. Contours are in steps of 0.5. Crosses show maximum likelihood estimates.}
\label{fig:lhood}
\end{figure}

A further diagnostic is presented in Figures~\ref{fig:WaveRW}~(a) and~(b), where ``fitted'' values of $\hat{S}=\max(\hat{A},\hat{B})$, and $\hat{V}=\hat{A}/(\hat{A}+\hat{B})$ are plotted for the pairs height--surge and period--surge, on a uniform scale. Plots for height--period are similar to those for period--surge, and hence are omitted. Here $(\hat{A},\hat{B}) = [\hat{F}_A^{-1}\{\tilde{F}_1(Z_1)\},\hat{F}_B^{-1}\{\tilde{F}_2(Z_2)\}]$, where $\hat{F}_A=\hat{F}_B$ is the fitted common pseudo-marginal distribution, and $\tilde{F}_1,\tilde{F}_2$ are the estimated true marginals. Points are plotted corresponding to $(\hat{S},\hat{V})$ where $\hat{S}$ exceeds its 90\% quantile. A lack of discernible patterns in Figures~\ref{fig:WaveRW}~(a) and (b) suggests that independence of $S$ and $V$ is a reasonable approximation. For comparison, Figures~\ref{fig:WaveRW}~(c) and~(d) show equivalent plots with $\max(X_P,Y_P)$ and $X_P/(X_P+Y_P)$ on a uniform scale, $(X_P,Y_P) = [\{1-\tilde{F}_1(Z_1)\}^{-1},\{1-\tilde{F}_2(Z_2)\}^{-1}]$; this would be the approach to modelling under asymptotic dependence \citep{ColesTawn91}. The patterns in Figure~\ref{fig:WaveRW}~(c) suggest that asymptotic dependence is plausible, but that a higher threshold is required for independence of $\max(X_P,Y_P)$ and $X_P/(X_P+Y_P)$. Figure~\ref{fig:WaveRW}~(d) shows that these variables would be dependent at any finite threshold.

\begin{figure}[p]
\centering
\includegraphics[width=0.25\textwidth]{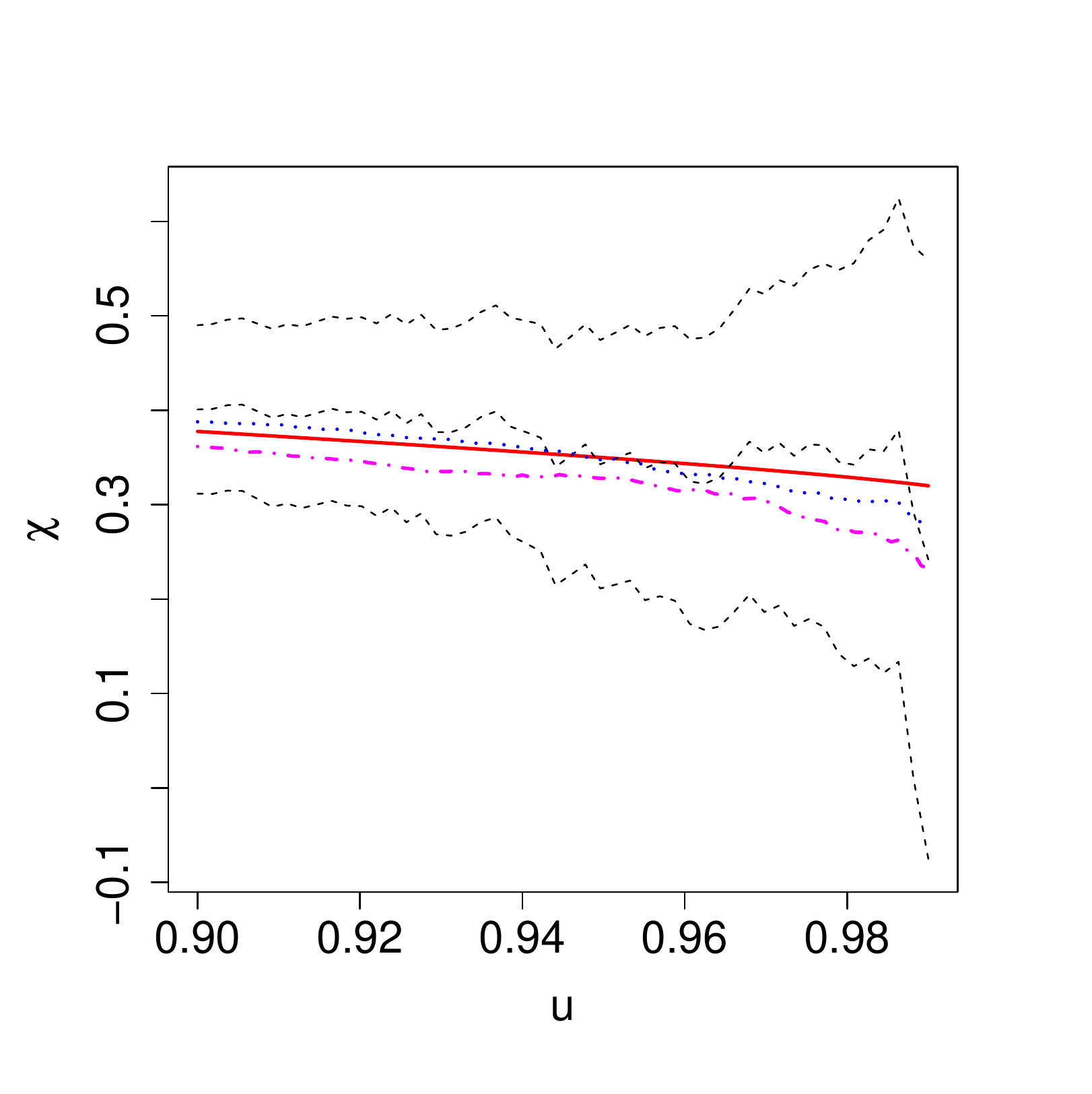}
\includegraphics[width=0.25\textwidth]{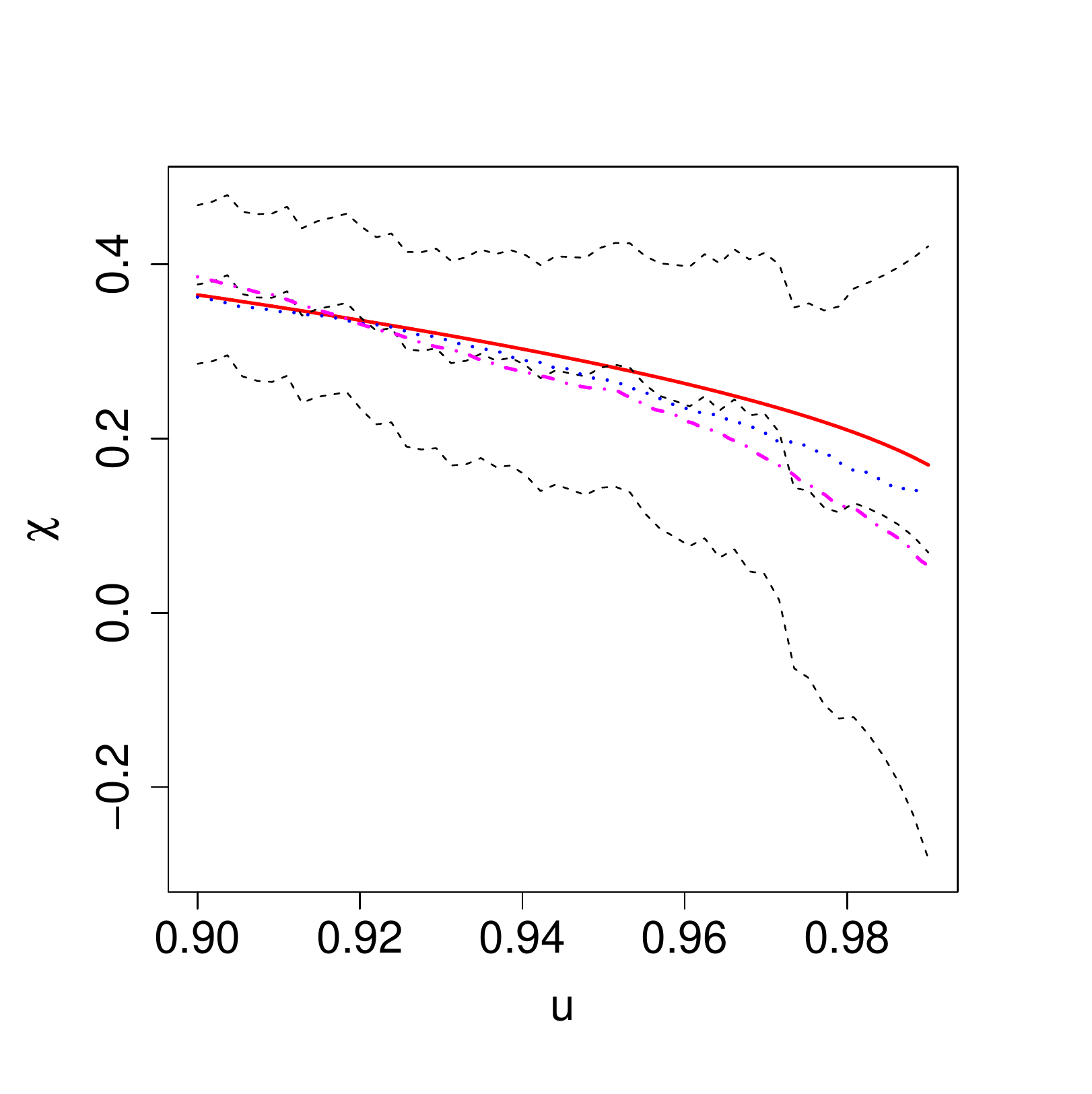}
\includegraphics[width=0.25\textwidth]{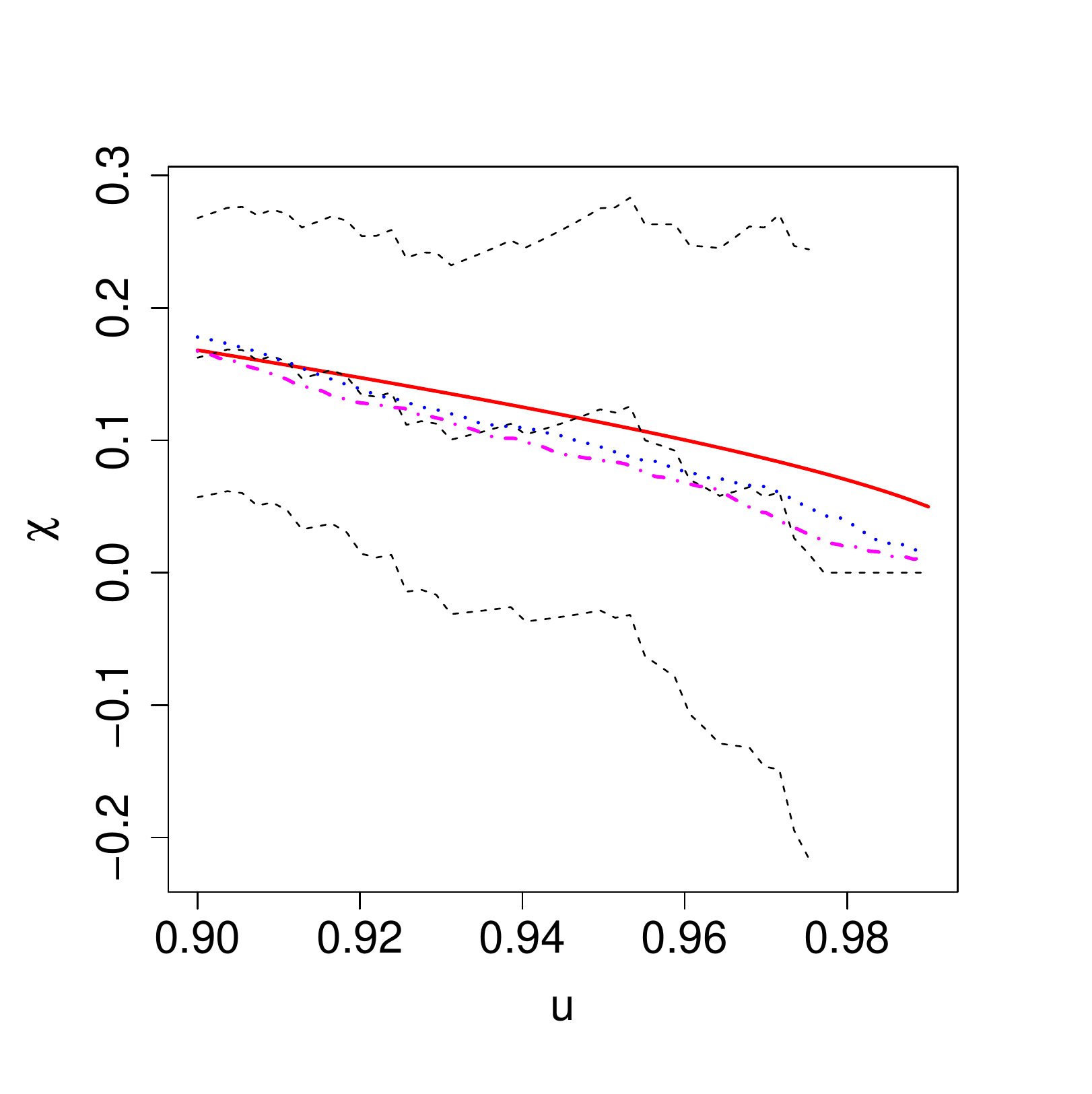}\\
\includegraphics[width=0.25\textwidth]{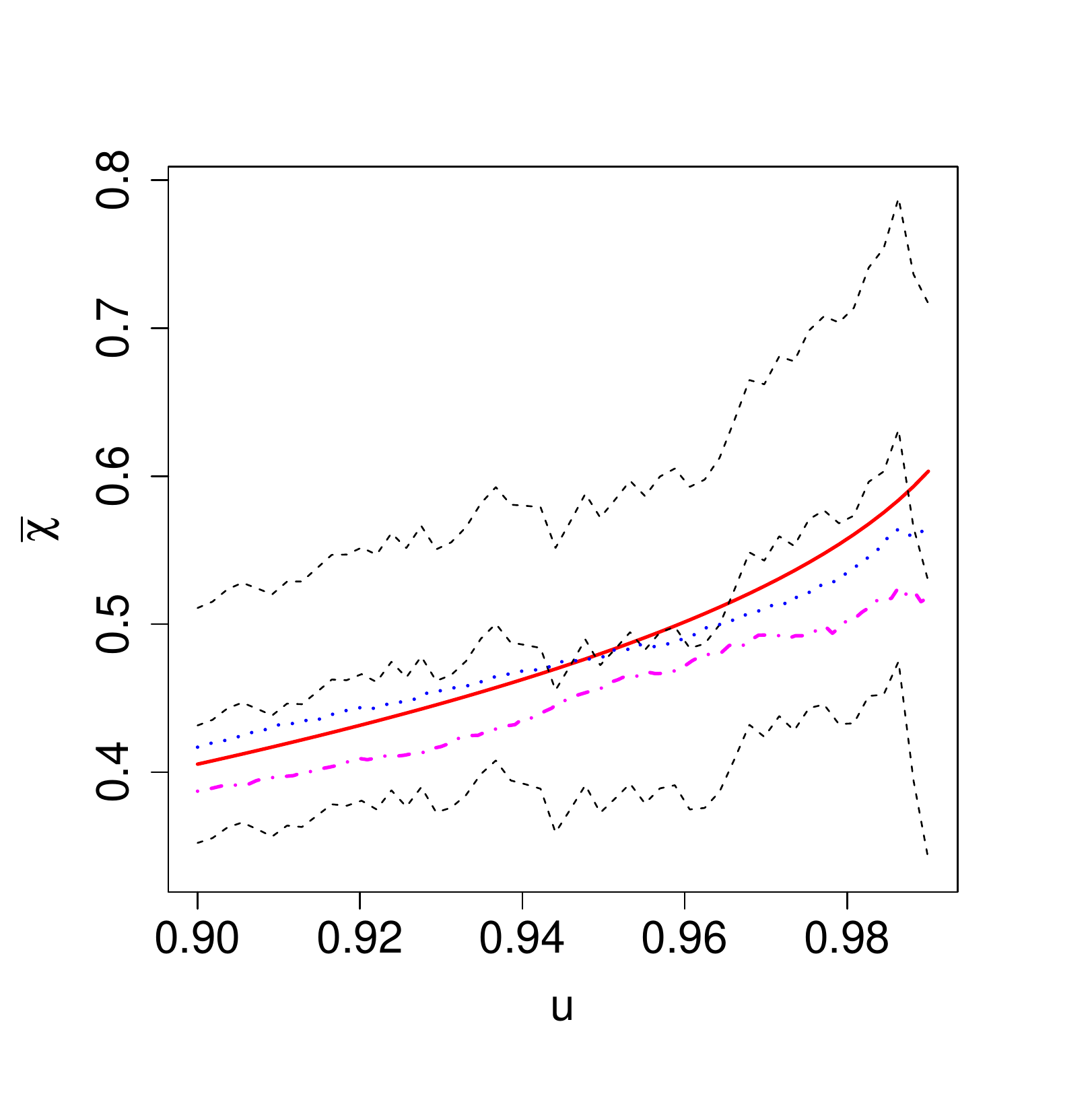}
\includegraphics[width=0.25\textwidth]{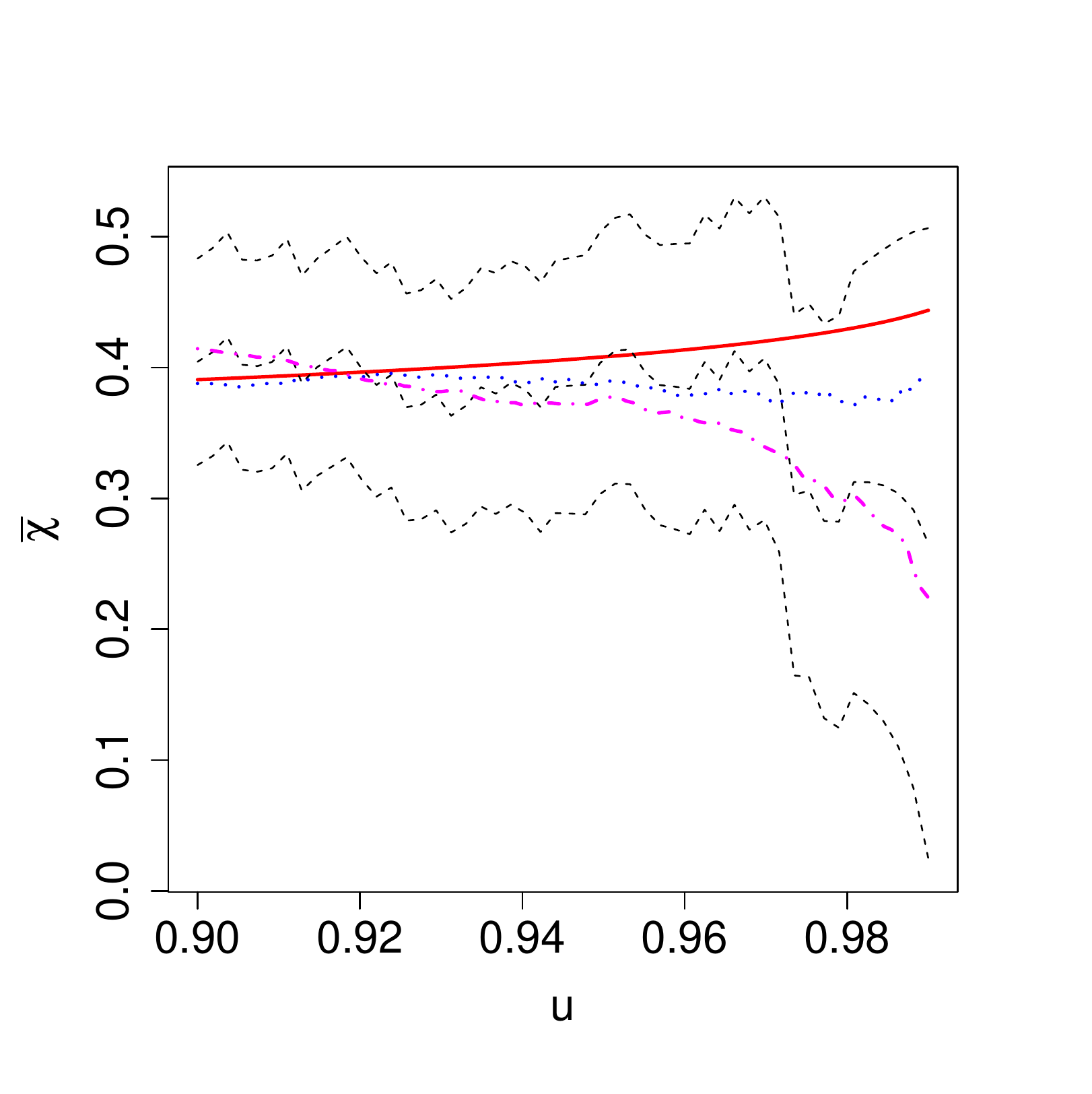}
\includegraphics[width=0.25\textwidth]{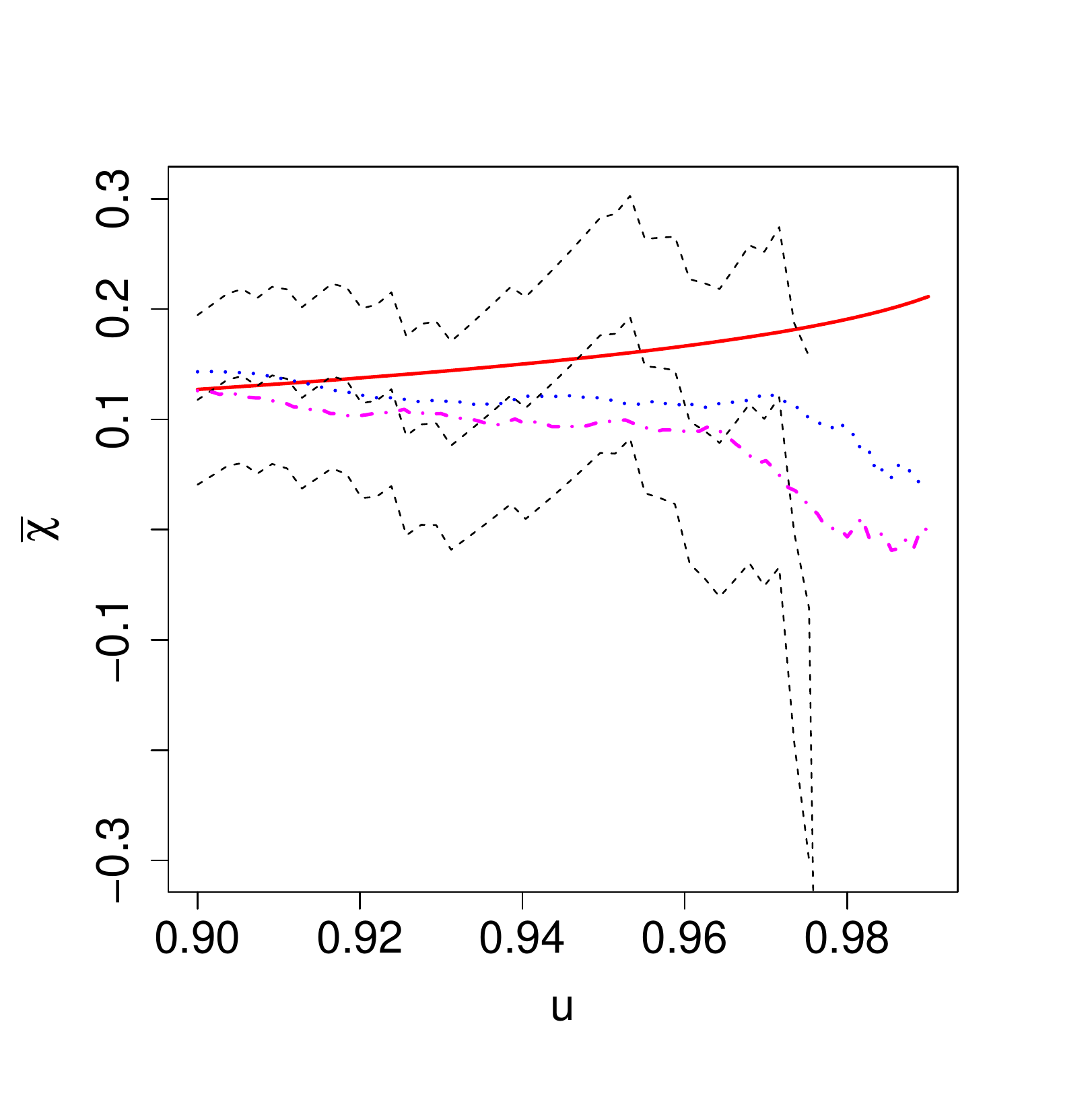}

\caption{Empirical (dashed lines, with approximate 95\% pointwise confidence intervals) and fitted (solid line) estimates of $\chi(u)$ (top row) and $\bar{\chi}(u)$ (bottom row), for $u\in(0.9,0.99)$. From left to right the pairs are height--surge, height--period and period--surge. Dot-dash lines: Heffernan--Tawn fit conditioning on the first variable of the pair. Dotted lines: Heffernan--Tawn fit conditioning on the second variable of the pair.}
\label{fig:wavechi}
\end{figure}

\begin{figure}[p]
\centering
\includegraphics[width=0.2\textwidth]{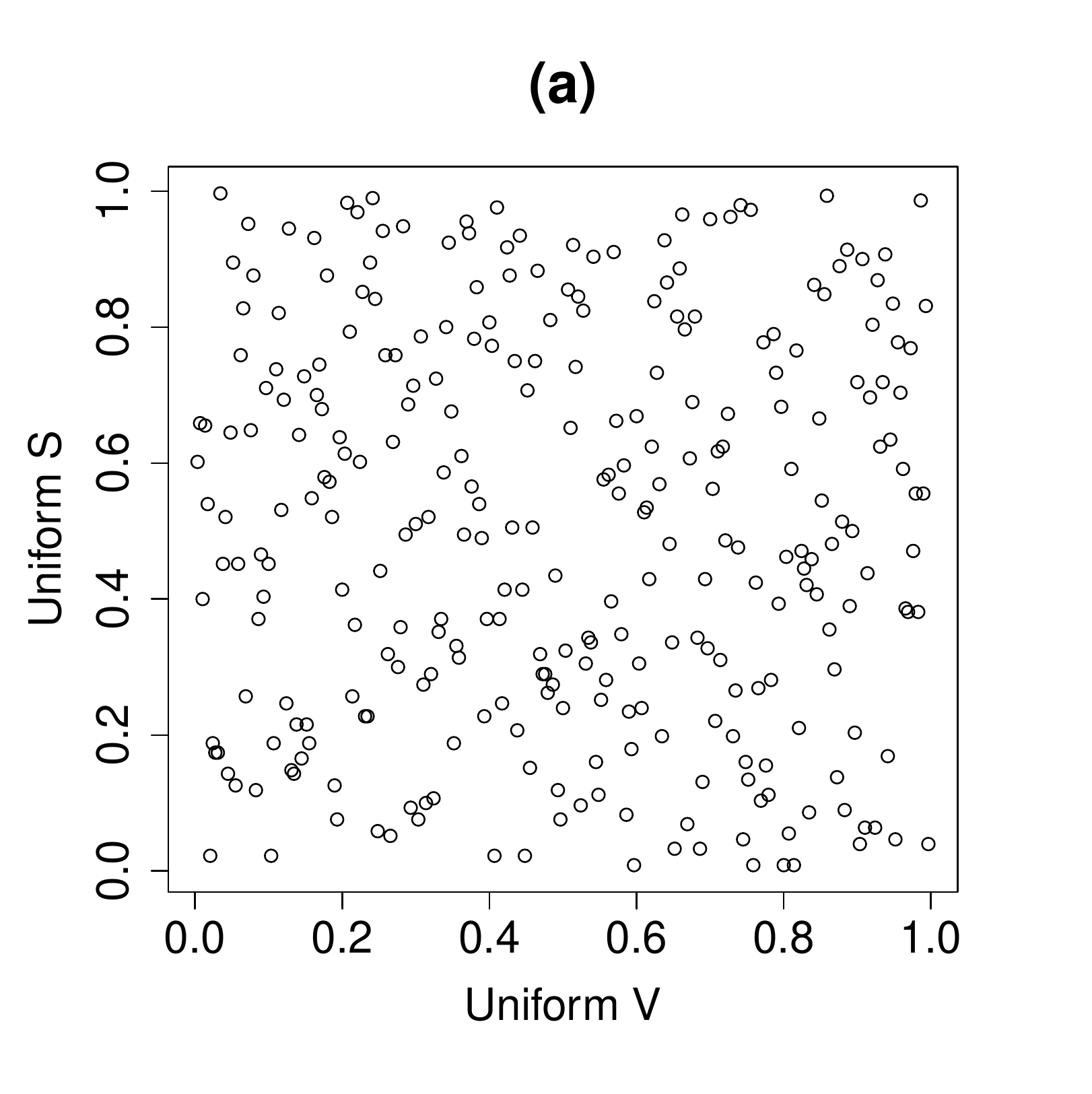}
\includegraphics[width=0.2\textwidth]{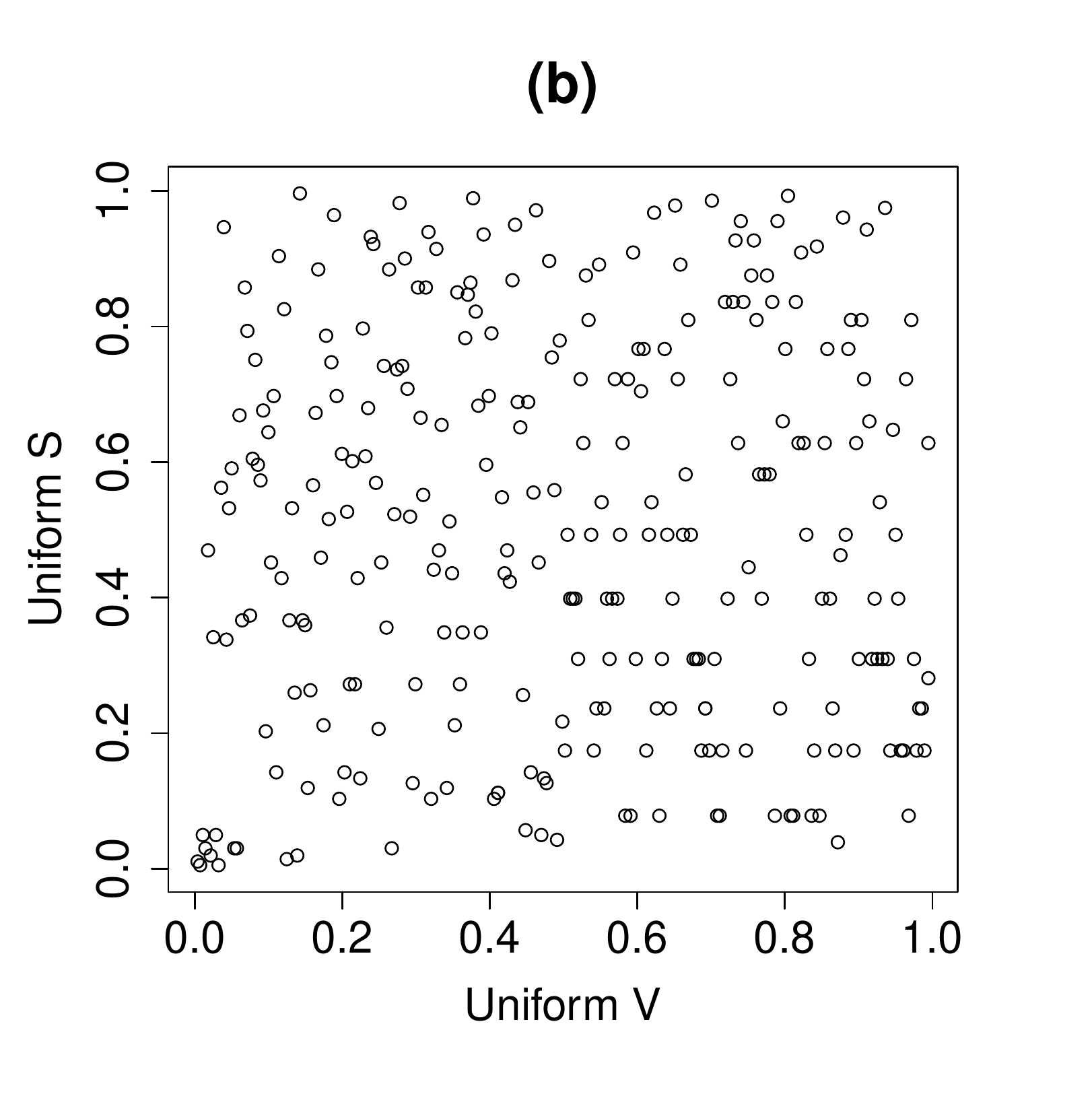}
\includegraphics[width=0.2\textwidth]{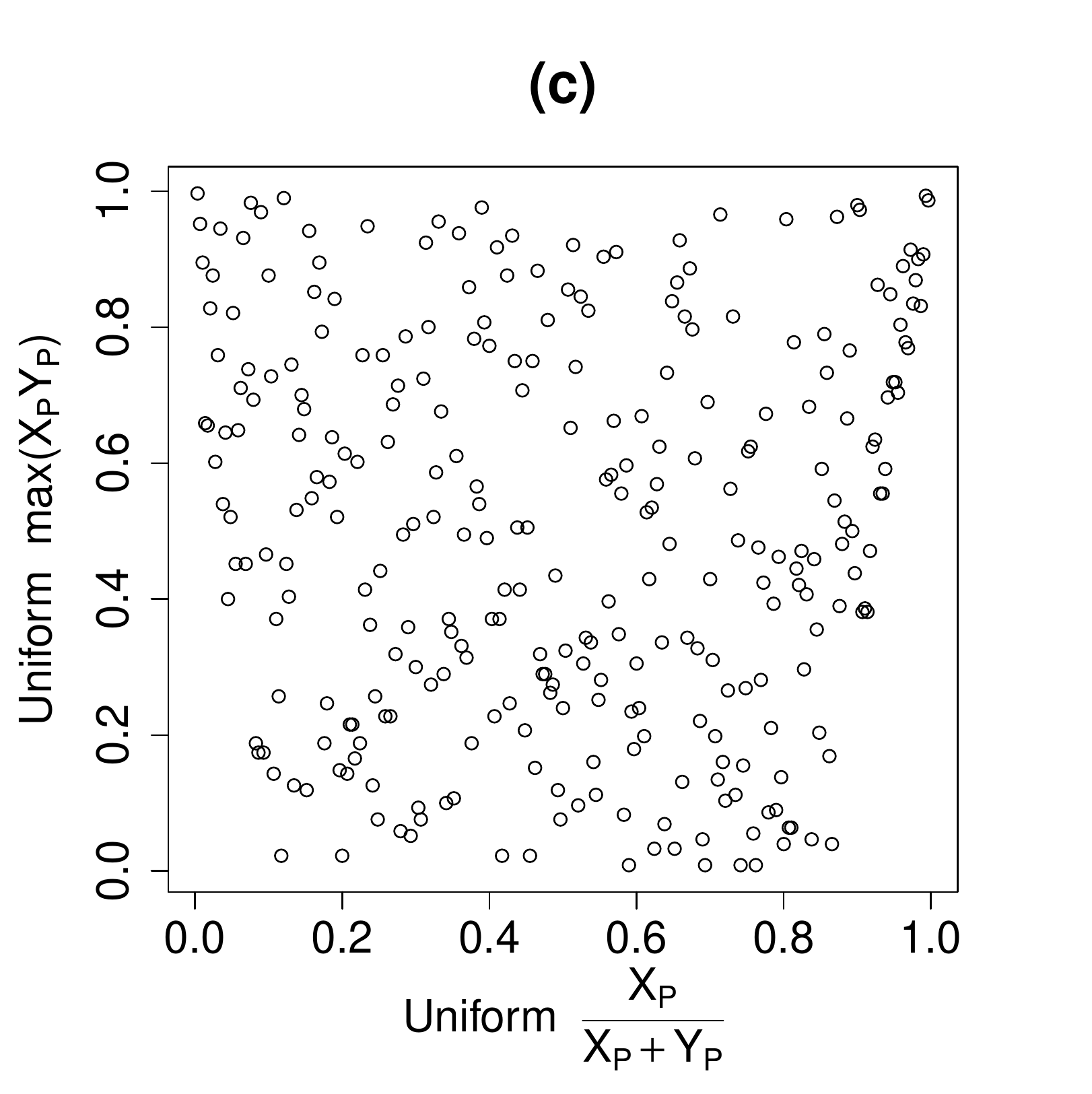}
\includegraphics[width=0.2\textwidth]{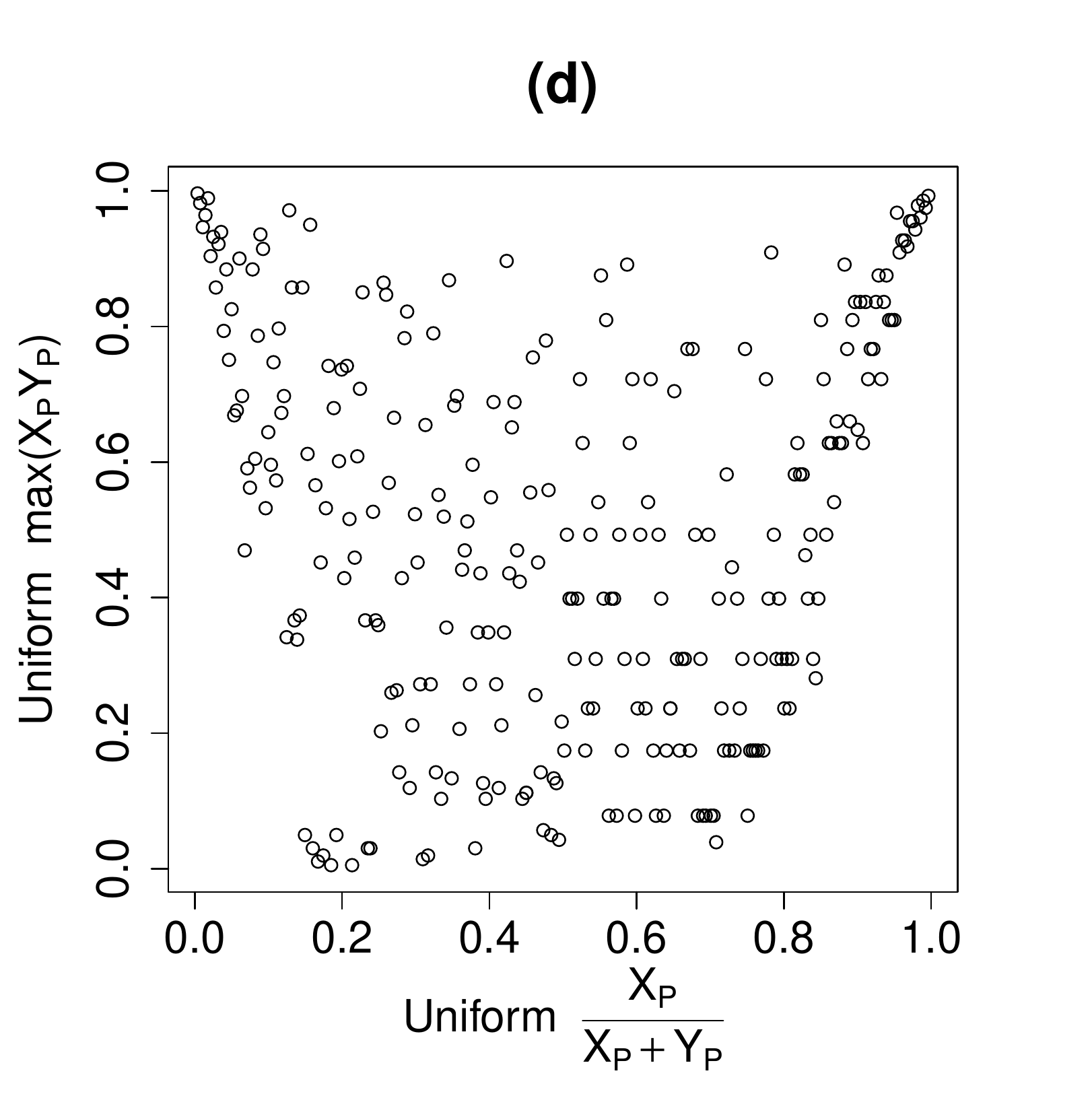}
\caption{Fitted $\hat{S}$ and $\hat{V}$, on a uniform scale, for (a) height--surge, and (b) period--surge. For comparison, $\max(X_P,Y_P)$ and $X_P/(X_P+Y_P)$ defined from the variables transformed to a standard Pareto scale are given for the same pairs in (c) and (d), respectively. Points which are aligned on the $\hat{S}$ / $\max(X_P,Y_P)$ axis are due to rounding of the data.}
\label{fig:WaveRW}
\end{figure}

\section{Extensions and discussion}
\label{sec:Discussion}

We have provided an alternative limit representation for bivariate extremes, which motivates a statistical model that can capture a wide spectrum of asymptotically dependent and asymptotically independent behaviour. An obvious question concerns extensions to higher dimensions. Assumption~\eqref{eq:RWgen} is indeed simple to extend to the multivariate case: in some common margins, $F$, the vector of positive random variables $\bm{X}=(X_1,\ldots,X_k) = [F^{-1}\{F_1(Z_1)\},\ldots, F^{-1}\{F_k(Z_k)\}]$ satisfies
\begin{align}
\lim_{t\to\infty} \Prob\left( \left. \frac{\bm{X}}{\sum_{i=1}^k X_i}\leq \bm{w}, \|\bm{X}\|_*> a(t)r + b(t) ~\right|~  \|\bm{X}\|_*>  b(t) \right) = J(\bm{w})\bar{K}(r),\quad r\geq 0, \label{eq:RWMV}
\end{align}
at continuity points of $J$, with $J$ placing mass on the interior of $\mathcal{S}_{k-1}^1 = \{\bm{w}\in\mathbb{R}^k_+ :\|\bm{w}\|_1 = 1\}$, and $\bar{K}$ as in~\eqref{eq:gp}. This is a more general assumption than multivariate regular variation, the $k$-dimensional extension of~\eqref{eq:RWAD}, that underpins much of classical multivariate extreme value theory \citep{deHaandeRonde98}.

However, the practical applicability of assumption~\eqref{eq:RWMV} in higher dimensions is more limited than in the bivariate case. The assumption that the distribution of $\bm{W}:=\bm{X}/\sum_{i=1}^k X_i$ has mass on the interior of $\mathcal{S}_{k-1}^1$ requires a certain regularity in the multivariate dependence structure, which is present in many theoretical examples, such as in the multivariate extensions of Examples~\ref{eg:mvrv}--\ref{eg:elliptical}, but often absent in datasets. For example, the data analyzed in Section~\ref{sec:Applications} exhibited asymptotic dependence between one pair of variables, but asymptotic independence between the other two pairs. The only existing model which can handle this is that of \citet{HeffernanTawn04}. However there are obvious issues with the curse of dimensionality when using a semiparametric model for higher dimensions. The simulation study in Section~\ref{sec:Inference} demonstrated a tendency for the semiparametric distribution estimator not to cover all parts of the plane, and this drawback would be exacerbated in higher dimensions.

We have assumed throughout that the radial variable $R = \|(X,Y)\|_*$ is defined by a norm, following the development of much of classical multivariate extreme value theory. In fact the convexity property does not appear necessary, and some recent articles on multivariate extremes have shifted focus on to positive homogeneous functions rather than norms \citep[e.g.][]{DombryRibatet15, SchefflerStoev15}. For our model the convexity property of $\|\cdot\|_m$ was used in some of the derivations; further work could explore more deeply the consequences of relaxing this assumption.

A simple extension to the practical modelling introduced in Sections~\ref{sec:Inference} and~\ref{sec:Applications} is to allow an asymmetric dependence structure. Our theoretical results in Section~\ref{sec:Model} already cover this scenario, but for simplicity of implementation we assumed the distribution of $V$ to be symmetric, so that the pseudo-marginals of $A,B$ were equal. As noted in Remark~\ref{rmk:1}, the implied $H$ incorporates the necessary moment constraint for any $F_V$.

In essence our approach is intermediate between assuming multivariate regular variation and the approach of \citet{HeffernanTawn04}. With the former, both the marginal distribution and the form of the normalization of each marginal variable, i.e., $\bm{X}_P/\|\bm{X}_P\|$, are fixed. This is restrictive, but allows for simpler characterization of the consequences of the assumption. With the latter, the margins are fixed to be of exponential type, but the form of the normalization of each marginal variable, $\{\bm{X}_E-\bm{b}(Y_E)\}/\bm{a}(Y_E)$, is not fixed. This permits great flexibility in the variety of distributions that satisfy the assumption, but leaves $k$ possible limits, each with $2(k-1)$ parameters to estimate, and a $(k-1)$-dimensional empirical distribution. Our main assumption does not fix the form of the margins, but does fix the form of the normalization of the variables $\bm{X}/\|\bm{X}\|_*$. This offers greater flexibility than multivariate regular variation, and although less flexible than the model of \citet{HeffernanTawn04} has the benefit of giving only a single limit. In the bivariate case, model~\eqref{eq:mod1}, inspired by~\eqref{eq:RWgen}, permits inference across both extremal dependence classes, with a smooth transition between them.

\subsubsection*{Acknowledgements}
This work was undertaken whilst JLW was based at EPFL and the University of Cambridge. We thank the Swiss National Science Foundation for funding, and the referees and associate editor for comments that have greatly improved the work.

\appendix
\section{Auxiliary results and proofs}
\label{sec:AuxiliaryResults}
\subsection{Link between~\eqref{eq:RWgen} and~\eqref{eq:dfcondn}}
\label{sec:EqConv}
\begin{prop}
 Let $W=X/(X+Y)$, $R=\|(X,Y)\|_*$, and assume that $W$ and $R$ have a joint density. Further assume $R$ to be in the domain of attraction of a generalized Pareto distribution, with normalization functions $a(t)>0$, $b(t)$. Then, provided that the limit on the right exists,
\begin{align*}
\lim_{t\to\infty} \Prob\left\{W\leq w, R > a(t)r + b(t) \mid R >  b(t) \right\} = J(w)\bar{K}(r)
\quad\Leftrightarrow\quad \lim_{t\to\infty} \Prob\{W\leq w \mid R = b(t)\} = J(w).
\end{align*}
\end{prop}
\begin{proof}
 Right to left: The statement on the right is equivalent to 
\begin{align*}
\lim_{t\to\infty} \frac{\frac{\partial}{\partial b(t)} \Prob\{W\leq w, R> b(t)\}}{\frac{\partial}{\partial b(t)} \Prob\{R>b(t)\}} = J(w).
\end{align*}
Since both $\lim_{t\to\infty} \Prob\{W\leq w, R>b(t)\}$ and $\lim_{t\to\infty} \Prob\{R>b(t)\}$ equal zero, but the ratio of the derivatives has limit $J(w)$, the general form of l'H\^{o}pital's rule states that
\begin{align*}
\lim_{t\to\infty} \frac{\Prob\{W\leq w, R>b(t)\}}{\Prob\{R>b(t)\}} = J(w).
\end{align*}
Consequently, as $t\to\infty$,
\begin{align*}
 \Prob\left\{\left. W\leq w, R > a(t)r + b(t) ~\right|~  R >  b(t) \right\} 
&= \frac{\Prob\{W\leq w, R > a(t)r + b(t)\}}{\Prob\{R > a(t)r + b(t)\}} \frac{\Prob\{R > a(t)r + b(t)\}}{\Prob\{R > b(t)\}}\rightarrow J(w)\bar{K}(r).
\end{align*}
Left to right: Set $r=0$ in the left-hand statement, yielding
\begin{align*}
\lim_{t\to\infty} \frac{\Prob\left\{ W\leq w, R > b(t) \right\}}{\Prob\{R >  b(t)\}} = J(w)\bar{K}(0),
\end{align*}
and note that $\bar{K}(0) = 1$. Then applying l'H\^{o}pital's rule again provides
\begin{align*}
\lim_{t\to\infty} \frac{\frac{\partial}{\partial b(t)} \Prob\{W\leq w, R>b(t)\}}{\frac{\partial}{\partial b(t)} \Prob\{R>b(t)\}} = J(w).
\end{align*}
\end{proof}

\subsection{Proofs of Propositions~\ref{l>0xyj:prop}--\ref{l<0xyj:prop}}
\label{sec:ProofProps}

We prove Propositions~\ref{l>0xyj:prop}--\ref{l<0xyj:prop}, giving the values of $\chi$, $\eta$ and $\kappa$ claimed in Section~\ref{sec:ExtremalDependence}. The following lemma on inversion of regularly varying functions will be useful throughout.
\begin{lem}
\label{geninv:lem}
Suppose $\Iind>0$ and $\Isvi$ is a slowly varying function such that $s\mapsto s^{-\Iind}\Isvi(s)$ defines a continuous strictly decreasing function from $[s_0,\infty)$ onto $(0,1]$ for some $s_0$. Then we can find a slowly varying function $\Isvo$ defined on $[1,\infty)$ such that $s^{-\Iind}\Isvi(s)=t^{-\Iind}$ whenever $s=t\Isvo^{1/\Iind}(t)$. Furthermore $\Isvo(t)\to c$ as $t\to\infty$ iff $\Isvi(s)\to c$ as $s\to\infty$ (here $c$ can be any value in the extended range $[0,+\infty]$).
\end{lem}
\noindent
The slowly varying functions $\Isvi^{-1/\Iind}$ and $\Isvo^{1/\Iind}$ are de Bruijn conjugates.
\begin{proof}
The expression $s\mapsto s\Isvi^{-1/\Iind}(s)$ defines a strictly increasing continuous map $[s_0,\infty)\to[1,\infty)$ which is regularly varying with index $1$ (note that $\Isvi^{-1/\Iind}$ is slowly varying). Let $\sigma:[1,\infty)\to[s_0,\infty)$ denote the corresponding inverse, which is also regularly varying with index $1$, and set $\Isvo(t)=t^{-\Iind}\sigma^{\Iind}(t)$ for all $t\ge1$; it follows that $\Isvo$ is continuous and slowly varying. Setting $s=\sigma(t)=t\Isvo^{1/\Iind}(t)$ we then get
\[
t=s\Isvi^{-1/\Iind}(s)=t\Isvo^{1/\Iind}(t)\,\Isvi^{-1/\Iind}\{t\Isvo^{1/\Iind}(t)\}
\ \Longrightarrow\ \Isvo(t)=\Isvi\{t\Isvo^{1/\Iind}(t)\}=\Isvi(s).
\]
The final part of the result follows (note that $t\Isvo^{1/\Iind}(t)\to\infty$ as $t\to\infty$ since $\Isvo$ is slowly varying).
\end{proof}

Define $\nf:[0,1]\to[1,\infty]$ as the reciprocal of $T:[0,1]\to[0,1]$ defined in Section~\ref{sec:ExtremalDependence}, i.e., $\nf(v)=\norm{(v,1-v)}/v$, so that $\nf(V) = 1/V_1$ and $\nf(1-V)=1/V_2$. Using this notation equation~\eqref{eq:marg} becomes
\begin{align}
 \Prob(A > x)&=\int_0^1 \{1+\lambda x\nf(v)\}_+^{-1/\lambda}\dsp F_V(v),& \Prob(B>y)&=\int_0^1 \{1+\lambda y\nf(1-v)\}_+^{-1/\lambda}\dsp F_V(v), \label{GPmargXY:eq1}
\end{align}
where the upper endpoint of the support is $\esupp=+\infty$ if $\lambda\ge0$ and $\esupp=-1/\lambda$ if $\lambda<0$; and~\eqref{eq:joint} becomes
\begin{subequations}
\begin{align}
\Prob(A > x,\,B > y)&=\int_0^1 \bigl[1+\lambda\max\{x\nf(v),\,y\nf(1-v)\}\bigr]_+^{-1/\lambda}\dsp F_V(v) \label{GPjointXY:eq1a}\\
&=\int_0^{x/(x+y)} \{1+\lambda x\nf(v)\}_+^{-1/\lambda}\dsp F_V(v) +\int_{x/(x+y)}^1 \{1+\lambda y\nf(1-v)\}_+^{-1/\lambda}\dsp F_V(v).\label{GPjointXY:eq1}
\end{align}
\end{subequations}
The expressions $x\mapsto\Prob(A > x)$ and $y\mapsto\Prob(B> y)$ define continuous strictly decreasing functions from $[0,\esupp)$ onto $(0,1]$; this observation can be used to help justify the conditions for Lemma~\ref{geninv:lem} when it is used below.

From Condition~\ref{Cond2}, $\nf(v)\ge(1-v)/v>1$ for $v<1/2$, while Conditions~\ref{Cond1} and~\ref{Cond3} imply $\nf(v)=1$ for some $v\in[1/2,1]$.   
Set $\muI=\{v\in[0,1]:\nf(v)=1\}$. Now $\nf(v)=\tilde{\tau}(1/v)$ where $\tilde{\tau}:[1,\infty]\to[1,\infty]$ is the continuous convex function defined by $\tilde{\tau}(u)=\norm{(1,u-1)}$. It follows that $\muI$ is a closed subinterval of $[1/2,1]$, so $\muI=[\lep,\rep]$ with $\lep$, $\rep$ as defined in Section~\ref{sec:ExtremalDependence}. Also note that $1/2\le\lep\le\rep\le1$,

\begin{equation}
\label{shapemu:eq}
\text{$\nf$ is strictly decreasing on $[0,\lep]$ and strictly increasing on $[\rep,1]$,}
\end{equation}
and
\begin{equation}
\label{muwwu1-wcomp:eq}
v\lessgtr\frac{x}{x+y}
\ \Longleftrightarrow\ yv\lessgtr x(1-v)
\ \Longleftrightarrow\ y\nf(1-v)\lessgtr x\nf(v).
\end{equation}

The quantities $\upm$ and $\lowm$ as given in Proposition~\ref{l<0xyj:prop} can be expressed
\[
\upm=\int_{\muI}\!\dsp F_V(v)=F_V(\rep)-F_V(\lep)
\ \ \text{and}\ \ 
\lowm=\int_{1-\muI}\!\dsp F_V(v)=F_V(1-\lep)-F_V(1-\rep);
\]
by Assumption~\ref{As_supp}, $\upm,\lowm>0$ iff $\lep\neq\rep$. We proceed with Cases~\ref{C:lplf}--\ref{C:lnli} in turn, firstly by establishing the form of the quantile functions $\xq$ and $\yq$, followed by proofs of the main Propositions concerning the behaviour of the joint survivor functions.

\subsubsection{Case 1: $\lambda>0$}
Recall the positive quantities $\upn, \lown$ defined in Remark~\ref{rmk:1}; these can be expressed $\upn=\lambda^{-1/\lambda}\int_0^1\nf^{-1/\lambda}(v)\dsp F_V(v)$, and $\lown=\lambda^{-1/\lambda}\int_0^1\nf^{-1/\lambda}(1-v)\dsp F_V(v)$.

\begin{prop}
\label{l>0xmarg:prop}
Let $\beta, \gamma>0$. Then there exist slowly varying functions $\ixmarg$, $\iymarg$ such that $\xq=t^{\lambda\beta}\ixmarg(t)$ and $\yq=t^{\lambda\gamma}\iymarg(t)$ for all $t\ge1$. Furthermore $\ixmarg(t)\to\upn^\lambda$ and $\iymarg(t)\to\lown^\lambda$ as $t\to\infty$.
\end{prop}

\begin{proof}
We have
\[
\xmarg(s):=s^{\beta}\Prob(A > s^{\lambda\beta})
=\int_0^1 \{s^{-\lambda\beta}+\lambda\nf(v)\}_+^{-1/\lambda}\dsp F_V(v).
\]
As $s$ increases from $0$ to $\infty$, $s^{-\lambda\beta}+\lambda\nf(v)$ decreases monotonically to $\lambda\nf(v)\ge\lambda$; hence $\{s^{-\lambda\beta}+\lambda\nf(v)\}_+^{-1/\lambda}$ increases monotonically to $\{\lambda\nf(v)\}^{-1/\lambda}\le\lambda^{-1/\lambda}$. Dominated convergence then gives
\[
\lim_{s\to\infty}\xmarg(s)=\int_0^1 \{\lambda\nf(v)\}^{-1/\lambda}\dsp F_V(v)=\upn.
\]
Since this limit is non-zero it follows that $\xmarg$ is slowly varying. The result for $\xq$ now follows from Lemma~\ref{geninv:lem} (with $\ixmarg=\Isvo^{\lambda}$). The $\yq$ case is similar.
\end{proof}

\begin{proof}[Proof of Proposition~\ref{l>0xyj:prop}]
Firstly suppose $\beta=\gamma$. From \eqref{GPjointXY:eq1a}, $\xyj$ as defined in Proposition~\ref{l>0xyj:prop}, is
\[
\xyj(t)=\int_0^1 \bigl[t^{-\lambda\beta}+\lambda\max\{\ixmarg(t)\nf(v),\,\iymarg(t)\nf(1-v)\}\bigr]_+^{-1/\lambda}\dsp F_V(v).
\]
Since $\nf\ge1$, and $\ixmarg(t)$ (or $\iymarg(t)$) has a non-zero limit as $t\to\infty$, we can bound $\max\{\ixmarg(t)\nf(v),\,\iymarg(t)\nf(1-v)\}$ uniformly away from $0$ for all sufficiently large $t$. Furthermore Proposition \ref{l>0xmarg:prop} implies $\max\{\ixmarg(t)\nf(v),\,\iymarg(t)\nf(1-v)\}\to\max\{\upn^\lambda\nf(v),\,\lown^\lambda\nf(1-v)\}$ as $t\to\infty$. Applying dominated convergence and using the definitions of $\upn$ and $\lown$ then gives
\[
\lim_{t\to\infty}\xyj(t)=\int_0^1\bigl[\lambda\max\{\upn^\lambda\nf(v),\,\lown^\lambda\nf(1-v)\}\bigr]^{-1/\lambda}\dsp F_V(v)=\chi_\lambda.
\]
The fact that this limit is non-zero implies $\xyj$ is slowly varying. Now assume $\beta<\gamma$ (the case $\beta>\gamma$ can be handled similarly). Then
\[
r(t):=\frac{\xq}{\xq+\yq}=\Bigl\{1+t^{\lambda(\gamma-\beta)}\frac{\iymarg(t)}{\ixmarg(t)}\Bigr\}^{-1}\to0,\quad t\to\infty.
\]
If $v\le r(t)$ then \eqref{muwwu1-wcomp:eq} gives
\begin{align*}
&\xq\nf(v)\ge \yq \nf(1-v)
\ \Longrightarrow\ 1 + \lambda\xq\nf(v)\ge 1 + \lambda\yq\nf(1- v)>\lambda\yq>0 \\
&\qquad\Longrightarrow\ 0<\bigl\{1+\lambda\yq\nf(1-v)\bigr\}_+^{-1/\lambda}-\bigl\{1+\lambda\xq\nf(v)\bigr\}_+^{-1/\lambda}\le\bigl\{\lambda\yq\bigr\}^{-1/\lambda}.
\end{align*}
Combined with \eqref{GPmargXY:eq1} and \eqref{GPjointXY:eq1} we thus have
\begin{align*}
0& \le\Prob\{B > \yq\}-\Prob\bigl\{A > \xq,B > \yq\bigr\} \\
 &=\int_0^{r(t)}\bigl[\bigl\{1+\lambda\yq\nf(1-v)\bigr\}_+^{-1/\lambda}-\bigl\{1+\lambda\xq\nf(v)\bigr\}_+^{-1/\lambda}\bigr]\dsp F_V(v)\\
&\le\int_0^{r(t)}\{\lambda\yq\}^{-1/\lambda}\dsp F_V(v)
= t^{-\gamma}\{\lambda\iymarg(t)\}^{-1/\lambda}F_V\{r(t)\}.
\end{align*}
The continuity of $F_V$ at $0$ gives $F_V\{r(t)\}\to F_V(0)=0$ as $t\to\infty$. Since $\Prob\{B>\yq\} = t^{-\gamma}$ we then get $\Prob\{A>\xq,B>\yq\}=t^{-\gamma}\{1+o(1)\}$ as $t\to\infty$. The result follows.
\end{proof}

\begin{proof}[Proof of Proposition~\ref{chi:prop}]
Note that by Condition~\ref{Cond3}, $\rep>1/2$. From~\eqref{eq:chi_lambda} we get $\chi_\lambda\le\mathcal{R}_-+\mathcal{R}_+$ where
\[
\mathcal{R}_-=\frac{\int_0^{1/2}\nf^{-1/\lambda}(v)\dsp F_V(v)}{\int_0^1\nf^{-1/\lambda}(v)\dsp F_V(v)}
\quad\text{and}\quad
\mathcal{R}_+=\frac{\int_{1/2}^1\nf^{-1/\lambda}(1-v)\dsp F_V(v)}{\int_0^1\nf^{-1/\lambda}(1-v)\dsp F_V(v)}.
\]
Now $\nf(v)\ge1$ with equality iff $v\in\muI$. Since $\muI\subseteq[1/2,1]$ dominated convergence then gives 
\[
\lim_{\lambda\to0^+}\int_0^{1/2}\!\nf^{-1/\lambda}(v)\dsp F_V(v)=0
\quad\text{and}\quad
\lim_{\lambda\to0^+}\int_0^1\!\nf^{-1/\lambda}(v)\dsp F_V(v)=\int_{\muI}\!\dsp F_V(v)=\upm.
\]
If $\lep=1/2<\rep$ then $\upm>0$ so $\mathcal{R}_-\to0$ as $\lambda\to0^+$. 
Otherwise $\lep>1/2$, in which case we can find $\delta>0$ so that $\nf(v)\ge1+\delta$ when $v\in[0,1/2]$. 
Setting $I_\delta=\{v\in[0,1]:\nf(v)\le1+\delta/2\}$ we then get
\[
\mathcal{R}_-\le\frac{\int_0^{1/2}(1+\delta)^{-1/\lambda}\dsp F_V(v)}{\int_{I_\delta}(1+\delta/2)^{-1/\lambda}\dsp F_V(v)}
\le\frac{(1+\delta)^{-1/\lambda}}{(1+\delta/2)^{-1/\lambda}\cm}
=\cm^{-1}\rho^{1/\lambda}
\]
where $\rho=1-\delta/(2+2\delta)\in(0,1)$ and $\cm:=\int_{I_\delta}\dsp F_V(v)>0$ (positivity follows from Assumption~\ref{As_supp} and the fact that the interval length $\abs{I_\delta}>0$). As $\lambda\to0^+$, $\rho^{1/\lambda}\to0$ and hence $\mathcal{R}_-\to0$. A similar argument shows $\mathcal{R}_+\to0$.
\end{proof}

\subsubsection{Case 2: $\lambda=0$}
Let $\beta,\gamma>0$ and set $\omega=\beta/(\beta+\gamma)\in(0,1)$. 
Then $\beta\nf(\omega)=\norm{(\beta,\gamma)}=\gamma\nf(1-\omega)$ while \eqref{muwwu1-wcomp:eq} gives
\begin{equation}
\label{nurelomega:eq}
\mnf(v):=\max\{\beta\nf(v),\gamma\nf(1-v)\}
=\begin{cases}
\beta\nf(v)&\text{if $0\le v\le\omega$,}\\
\gamma\nf(1-v)&\text{if $\omega\le v\le1$.}
\end{cases}
\end{equation}
The function $\mnf$ is a positive, continuous and convex function on $[0,1]$, with $\mnf(0)=+\infty=\mnf(1)$. Set $\widehat{\nu}:=\min\{\mnf(v):v\in[0,1]\}$ and $\nuI:=\{v\in[0,1]:\mnf(v)=\widehat{\nu}\}$; in particular, $\nuI$ is a non-empty closed subinterval of $[0,1]$. The general shape of $\mnf$ and key properties of $\widehat{\nu}$ and $\nuI$ can be deduced from \eqref{shapemu:eq}:
\begin{description}
\item{\emph{C1: $\omega\in[1-\lep,\lep{]}$.}}
Then $\beta\nf(v)$ is strictly decreasing on $[0,\omega]$, $\gamma\nf(1-v)$ is strictly increasing on $[\omega,1]$ and these quantities are equal when $v=\omega$. It follows that $\nuI=\{\omega\}$ and $\widehat{\nu}=\beta\nf(\omega)=\gamma\nf(1-\omega)=\norm{(\beta,\gamma)}$. 
\item{\emph{C2: $\omega\in(\lep,1)$.}}
Then $\mnf(v)=\beta\nf(v)$ is strictly decreasing on $[0,\lep]$ and $\mnf(v)=\beta\nf(v)=\beta$ (a constant) on $[\lep,\min\{\omega,\rep\}]$. Also $\beta\nf(v)$ is strictly increasing on $[\rep,1]$ and $\omega>\lep\ge1-\lep$ so $\gamma\nf(1-v)$ is strictly increasing and not less than $\beta\nf(v)$ on $[\omega,1]$; hence $\mnf(v)=\max\{\beta\nf(v),\gamma\nf(1-v)\}$ is strictly increasing on $[\min\{\omega,\rep\},1]$. It follows that $\nuI=[\lep,\min\{\omega,\rep\}]$ and $\widehat{\nu}=\beta=\norm[\infty]{(\beta,\gamma)}$ (note that, $\omega>\lep\ge1/2$ which implies $\beta>\gamma$). 
\item{\emph{C3: $\omega\in(0,1-\lep)$.}}
By a similar argument to \emph{C2}, $\nuI=[\max\{\omega,1-\rep\},1-\lep]$ and $\widehat{\nu}=\gamma=\norm[\infty]{(\beta,\gamma)}$. 
\end{description}
The main results in this case are built from the following lemma.
\begin{lem}
\label{intregvar:lem}
Suppose $\Lind:[0,1]\to[0,\infty]$ is continuous,
$\Lrvi$ is regularly varying at infinity with index $\Lii>0$,
and $\Lintv,\Lintv[s]\subseteq[0,1]$ for $s\ge0$ is a collection of closed intervals with the interval length $\abs{\Lintv}>0$ and $\Lintv[s]\to\Lintv$ as $s\to\infty$.
Define $\Lrvo$ by
\[
\Lrvo(s)=\int_{\Lintv[s]} \Lrvi^{-\Lind(v)}(s)\dsp F_V(v)
\]
for each $s\ge0$, and set $\LimA=\min\{\Lind(v):v\in\Lintv\}$. Then $\Lrvo$ is regularly varying with index $-\LimA\Lii$. 
\end{lem}
Note that by $\Lintv[s]\to\Lintv$ we mean that the Hausdorff distance between $\Lintv[s]$ and $\Lintv$ tends to $0$; equivalently, the end points of $\Lintv[s]$ converge to the end points of $\Lintv$.

%----

\begin{proof}
For each $\delta>0$ set $J_\delta=\{v\in[0,1]:\Lind(v)\le\LimA+\delta\}$. 

\smallskip

\noindent
\emph{Claim 1: there exists $S_{1,\delta}$ such that $\abs{\Lind(v)-\LimA}\le\delta$ when $s\ge S_{1,\delta}$ and \mbox{$v\in\Lintv[s]\cap J_\delta$.}} The continuity of $\Lind$ implies $U:=\{v\in[0,1]:\Lind(v)>\alpha-\delta\}$ is an open neighbourhood of $\Lintv\cap J_\delta\neq\emptyset$. Since $\Lintv[s]\to\Lintv$ as $s\to\infty$ it follows that $\Lintv[s]\cap J_\delta\subseteq U$ for all sufficiently large $s$. 

\smallskip

\noindent
\emph{Claim 2: there exists $S_{2,\delta}$ and $\cm>0$ such that $\int_{\Lintv[s]\cap J_{\delta/4}}\dsp F_V(v)\ge C_\delta$ for all $s\ge S_{2,\delta}$.} Choose $\tilde{v}\in\Lintv$ and $\delta_0>0$ so that $\Lind(\tilde{v})=\LimA$ and $J':=[\tilde{v}-\delta_0,\tilde{v}+\delta_0]\subseteq J_{\delta/4}$. Then $\Lintv\cap J'$ is an interval of length at least $\delta_1=\min(\delta_0,\abs{\Lintv})>0$ (recall that $\Lintv$ is an interval). Since $\Lintv[s]$ is an interval converging to $\Lintv$ it follows that, for all sufficiently large $s$, $\Lintv[s]\cap J'$ is an interval of length at least $\delta_1/2$, which is contained in $\Lintv[s]\cap J_{\delta/4}$. 
We can then let $\cm$ be the infimum of $\int_K\dsp F_V(v)$, taken over all intervals $K\subseteq[0,1]$ of length at least $\delta_1/2$; this quantity is positive by Assumption~1.

\smallskip

Setting
\[
\Lrvo[\delta](s)=\int_{\Lintv[s]\cap J_\delta}\Lrvi^{-\Lind(v)}(s)\dsp F_V(v)
\quad\text{and}\quad
\Lcrvo[\delta](s)=\int_{\Lintv[s]\setminus J_\delta}\Lrvi^{-\Lind(v)}(s)\dsp F_V(v)
\]
we clearly have
\begin{equation}
\label{splitxmarg:eq}
\Lrvo(s)=\Lrvo[\delta](s)+\Lcrvo[\delta](s).
\end{equation}

\smallskip

\noindent
\emph{Claim 3: there exists $S_{3,\delta}$ such that}
\begin{equation}
\label{margrat2:eq}
1\le\frac{\Lrvo(s)}{\Lrvo[\delta](s)}
%=1+\frac{\Lcrvo[\delta](s)}{\Lrvo[\delta](s)}
\le1+\cm^{-1}s^{-\Lii\delta/4}
\quad\text{\emph{for $s\ge S_{3,\delta}$.}}
\end{equation}
Set $\sigma=\Lii\delta/\{4(\LimA+\delta)\}\in(0,\Lii/4]$. Since $\Lrvi$ is regularly varying with index $\Lii$ there exists $S_{3,\delta}'\ge1$ such that
\[
s^{\Lii-\sigma}\le \Lrvi(s)\le s^{\Lii+\sigma}\quad\text{for $s\ge S_{3,\delta}'$.}
\]
If $v\in J_{\delta/4}$ then $\Lind(v)\le\LimA+\delta/4$ so
\[
\Lind(v)(\Lii+\sigma)
\le\LimA\Lii+\sigma(\LimA+\delta/4)+\Lii\delta/4
\le\LimA\Lii+\sigma(\LimA+\delta)+\Lii\delta/4
=\LimA\Lii+\Lii\delta/2
\]
so, for any $s\ge S_{3,\delta}'$, 
\[
\Lrvi^{-\Lind(v)}(s)\ge s^{-\Lind(v)(\Lii+\sigma)}\ge s^{-\LimA\Lii-\Lii\delta/2}.
\]
When $s\ge\max\{S_{2,\delta},S_{3,\delta}'\}$, Claim 2 then leads to
\[
\Lrvo[\delta](s)
\ge\Lrvo[\delta/4](s)=\int_{\Lintv[s]\cap J_{\delta/4}}\Lrvi^{-\Lind(v)}(s)\dsp F_V(v)
\ge s^{-\LimA\Lii-\Lii\delta/2}\int_{\Lintv[s]\cap J_{\delta/4}}\dsp F_V(v)\ge\cm s^{-\LimA\Lii-\Lii\delta/2}.
\]
On the other hand, if $v\notin J_\delta$ then $\Lind(v)\ge\LimA+\delta$ so
\[
\Lind(v)(\Lii-\sigma)\ge(\LimA+\delta)(\Lii-\sigma)=\LimA\Lii-\sigma(\LimA+\delta)+\Lii\delta=\LimA\Lii+3\Lii\delta/4,
\]
and thus, for any $s\ge S_{3,\delta}'$, 
\[
\Lrvi^{-\Lind(v)}(s)\le s^{-\Lind(v)(\Lii-\sigma)}\le s^{-\LimA\Lii-3\Lii\delta/4}.
\]
When $s\ge S_{3,\delta}'$ it follows that
\[
\Lcrvo[\delta](s)
=\int_{\Lintv[s]\setminus J_\delta}\Lrvi^{-\Lind(v)}(s)\dsp F_V(v)
\le s^{-\LimA\Lii-3\Lii\delta/4}\int_{\Lintv[s]\setminus J_\delta}\dsp F_V(v)\le s^{-\LimA\Lii-3\Lii\delta/4}.
\]
When $s\ge\max(S_{2,\delta},S_{3,\delta}')$ our estimates for $\Lrvo[\delta](s)$ and $\Lcrvo[\delta](s)$ can be combined with \eqref{splitxmarg:eq} to give \eqref{margrat2:eq}.

\medskip

Let $l\ge1$ and $\epsilon>0$. Choose $\delta\in(0,1]$ so that $(1+\delta)^{\LimA+\delta}l^{\Lii\delta}\le1+\epsilon$.
Since $u$ is regularly varying with index $\Lii$ we can find $S_{4,\delta}$ such that
\[
(1+\delta)^{-1}l^{\Lii}\le\frac{\Lrvi(ls)}{\Lrvi(s)}\le(1+\delta)l^{\Lii}
\quad\text{for $s\ge S_{4,\delta}$.}
\]
If $v\in\Lintv[s]\cap J_{\delta}$ and $s\ge\max\{S_{1,\delta},S_{4,\delta}\}$, Claim 1 gives $\LimA-\delta\le\Lind(v)\le\LimA+\delta$ and so
\begin{align*}
&(1+\epsilon)^{-1}l^{-\LimA\Lii}
\le(1+\delta)^{-(\LimA+\delta)}l^{-(\LimA+\delta)\Lii}
\le(1+\delta)^{-\Lind(v)}l^{-\Lind(v)\Lii}\\
&\qquad{}\le\frac{\Lrvi^{-\Lind(v)}(ls)}{\Lrvi^{-\Lind(v)}(s)}
\le(1+\delta)^{\Lind(v)}l^{-\Lind(v)\Lii}
\le(1+\delta)^{\LimA+\delta}l^{-(\LimA-\delta)\Lii}
\le(1+\epsilon)l^{-\LimA\Lii}.
\end{align*}
Integration then gives
\begin{equation}
\label{margrat1:eq}
\frac{\Lrvo[\delta](ls)}{\Lrvo[\delta](s)}\in[(1+\epsilon)^{-1}l^{-\LimA\Lii},(1+\epsilon)l^{-\LimA\Lii}].
\end{equation}

Choose $S\ge\max\{S_{1,\delta},\dots,S_{4,\delta}\}$ so that $S^{-\Lii\delta/4}\le\cm\epsilon$. Now
\[
\frac{\Lrvo(ls)}{\Lrvo(s)}
=\frac{\Lrvo(ls)}{\Lrvo[\delta](ls)}\,\frac{\Lrvo[\delta](ls)}{\Lrvo[\delta](s)}\,\frac{\Lrvo[\delta](s)}{\Lrvo(s)}.
\]
For $s\ge S$ the middle term on the right hand side belongs to $[(1+\epsilon)^{-1}l^{-\LimA\Lii},(1+\epsilon)l^{-\LimA\Lii}]$ by \eqref{margrat1:eq}, while the first and third terms belong to $[1,1+\epsilon]$ and $[(1+\epsilon)^{-1},1]$ respectively by \eqref{margrat2:eq} (note that, $l\ge1$ so $ls\ge s\ge S$). Thus $\Lrvo(ls)/\Lrvo(s)\in[(1+\epsilon)^{-2}l^{-\LimA\Lii},(1+\epsilon)^2l^{-\LimA\Lii}]$ for any $s\ge S$. Since $\epsilon>0$ was arbitrary it follows that $\Lrvo(ls)/\Lrvo(s)\to l^{-\LimA\Lii}$ as $s\to\infty$; hence $\Lrvo$ is regularly varying with index $-\LimA\Lii$. 
\end{proof}

%----

\begin{prop}
\label{xmarg:prop}
Let $\beta, \gamma>0$. Then there exist slowly varying functions $\ixmarg$, $\iymarg$ such that $\xq=\log\{t^\beta\ixmarg(t)\}$ and $\yq=\log\{t^\gamma\iymarg(t)\}$ for all $t\ge1$. Furthermore $\ixmarg$, $\iymarg$ are continuous, take values in $[\upm,1]$ and $[\lowm,1]$ respectively, and satisfy $\ixmarg(t)\to\upm$ and $\iymarg(t)\to\lowm$ as $t\to\infty$.
\end{prop}

\begin{proof}
For $s\ge1$, using \eqref{GPmargXY:eq1},
\[
\xmarg(s):=s^\beta\Prob(A > \beta\log s)
=s^\beta\int_0^1 e^{-\beta\nf(v)\log s}\dsp F_V(v)
=\int_0^1 s^{-\beta\{\nf(v)-1\}}\dsp F_V(v).
\]
Now $\beta\{\nf(v)-1\}\ge0$ with equality iff $v\in \muI$. Dominated convergence then gives
\[
\lim_{s\to\infty}\xmarg(s)=\int_0^1\lim_{s\to\infty}s^{-\beta\{\nf(v)-1\}}\dsp F_V(v)=\int_{\muI}\dsp F_V(v)=\upm.
\]
By Lemma \ref{intregvar:lem} we know that $\xmarg$ is slowly varying. 
The result for $\xq$ now follows from Lemma \ref{geninv:lem} (with $\ixmarg=\Isvo$). The $\yq$ case is similar.
\end{proof}

\begin{proof}[Proof of Proposition~\ref{xyj:prop}]
Setting
\begin{equation}
\label{l0rt:eq}
r(t)=\frac{\xq}{\xq+\yq}=\frac{\beta\log t+\log\ixmarg(t)}{(\beta+\gamma)\log t+\log\ixmarg(t)\iymarg(t)}
\end{equation}
we have $r(t)\to\omega$ as $t\to\infty$ (note that $\ixmarg$ and $\iymarg$ are slowly varying).
Furthermore \eqref{GPjointXY:eq1} gives
\begin{align}
\Prob\bigl\{A>\xq,\,B>\yq\bigr\} &=\int_0^{r(t)}e^{-\nf(v)\log\{t^\beta\ixmarg(t)\}}\dsp F_V(v)+\int_{r(t)}^1 e^{-\nf(1-v)\log\{t^\gamma\iymarg(t)\}}\dsp F_V(v)\nonumber\\
\label{jointXY:eq2}
&=\int_0^{r(t)}\bigl\{t^\beta\ixmarg(t)\bigr\}^{-\nf(v)}\dsp F_V(v)+\int_{r(t)}^1\bigl\{t^\gamma\iymarg(t)\bigr\}^{-\nf(1-v)}\dsp F_V(v).
\end{align}
Now assume $\beta\le\gamma$ (the case $\beta\ge\gamma$ can be handled similarly).
Then $\omega\le1/2\le\lep$ so \eqref{shapemu:eq} gives $\min\{\nf(v):v\in[0,\omega]\}=\nf(\omega)=\norm{(\beta,\gamma)}/\beta$. 
Furthermore, $\nuI\subseteq[\omega,1]$ (recall the description of $\mnf$ at the beginning of this section) so $\min\{\nf(1-v):v\in[\omega,1]\}
=\gamma^{-1}\min\{\mnf(v):v\in[\omega,1]\}
=\widehat{\nu}/\gamma$. Lemma \ref{intregvar:lem} can now be applied to show that the integrals on the right hand side of \eqref{jointXY:eq2} are regularly varying functions, the first with index $-\norm{(\beta,\gamma)}\le-\widehat{\nu}$ and the second with index $-\widehat{\nu}$. By the forms of $\widehat{\nu}$ described in \emph{C1--C3} immediately preceding Lemma~\ref{intregvar:lem}, the result follows.
\end{proof}

The fact that $\chi=0$ when $\eta=1$ in this case is given by the following.
\begin{prop}
\label{prop:lam0chi}
If $\lep\neq\rep$ (equivalently $\upm,\lowm>0$) and $1-\lep\le\omega\le\lep$ then $\lim_{t\to\infty}\xyj(t)= 0$.
\end{prop}

\begin{proof}
From \eqref{GPjointXY:eq1a} and Proposition~\ref{xmarg:prop} we have
\begin{align*}
\Prob\{A>\xq,\,B>\yq\} &=\int_0^1\min\bigl[e^{-\nf(v)\log\{t^\beta\ixmarg(t)\}},\,e^{-\nf(1-v)\log\{t^\gamma\iymarg(t)\}}\bigr]\dsp F_V(v).
\end{align*}
By Proposition~\ref{xyj:prop} we then get $\xyj(t)=\int_0^1\jsvf(t)\dsp F_V(v)$ where
\[
\jsvf(t)=t^{\widehat{\nu}}\min\big\{t^{-\beta\nf(v)}\ixmarg^{-\nf(v)}(t),\,t^{-\gamma\nf(1-v)}\iymarg^{-\nf(1-v)}(t)\bigr\}.
\]
Now $\nf\ge1$ so $\ixmarg^{-\nf(v)}(t),\,\iymarg^{-\nf(1-v)}(t)\le C=\max\{\upm^{-1},\lowm^{-1}\}$ using Proposition \ref{xmarg:prop}.
Furthermore $\widehat{\nu}\le\max\{\beta\nf(v),\,\gamma\nf(1-v)\}$ (by definition) leading to $\jsvf(t)\le C$ for all $v$ and $t\ge1$. If $v\notin\nuI$ then $\widehat{\nu}<\max\{\beta\nf(v),\,\gamma\nf(1-v)\}$ so $\jsvf(t)\to0$ as $t\to\infty$. In particular, if $\omega\in[1-\lep,\lep]$ it follows that $\nuI=\{\omega\}$ and hence $\jsvf(t)\to0$ as $t\to\infty$ whenever $v\neq\omega$; dominated convergence then gives $\lim_{t\to\infty}\xyj(t)=0$.
\end{proof}

\subsubsection{Case 3: $\lambda<0,$ $\norm{(1,1)}=\|(1,1)\|_\infty$, with Assumption~\ref{As_cd}}

\begin{prop}
\label{l<0xmarg:prop}
Let $\beta,\gamma>0$. Then there exist slowly varying functions $\ixmarg$, $\iymarg$ such that $\xq=\esupp-t^{\lambda\beta}\ixmarg(t)$ and $\yq=\esupp-t^{\lambda\gamma}\iymarg(t)$ for all $t\ge1$. Furthermore $\ixmarg(t)\to\esupp\upm^{\lambda}$ and $\iymarg(t)\to\esupp\lowm^{\lambda}$ as $t\to\infty$.
\end{prop}

\begin{proof}
Set $S_0=\esupp^{1/(\lambda\beta)}$. For $s\ge S_0$ we get
\begin{align}
\xmarg(s):=s^{\beta}\Prob\bigl(A > \esupp-s^{\lambda\beta}\bigr) &=\int_0^1 \bigl[s^{-\lambda\beta}\bigl\{1-\lambda(1/\lambda+s^{\lambda\beta})\nf(v)\bigr\}\bigr]_+^{-1/\lambda}\dsp F_V(v)\nonumber\\
\label{l<0xmargin:eq1}
&=\int_0^1 \bigl[(s^{-\lambda\beta}+\lambda)\{1-\nf(v)\}-\lambda\bigr]_+^{-1/\lambda}\dsp F_V(v),
\end{align}
using \eqref{GPmargXY:eq1}.
For $s\ge S_0$ we have $(s^{-\lambda\beta}+\lambda)\{1-\nf(v)\}\le0$ (recall that $\nf(v)\ge1$) 
so the integrand in \eqref{l<0xmargin:eq1} is bounded above by $(-\lambda)^{-1/\lambda}$.
Also note that $s^{-\lambda\beta}\to+\infty$ as $s\to\infty$, so 
\[
\lim_{s\to\infty}\bigl[(s^{-\lambda\beta}+\lambda)\{1-\nf(v)\}-\lambda\bigr]_+
=\begin{cases}
0&\text{if $\nf(v)>1$,}\\
-\lambda&\text{if $\nf(v)=1$.} 
\end{cases}
\]
As $\{v:\nf(v)=1\}=\muI$, dominated convergence now gives
\[
\lim_{s\to\infty}\xmarg(s)=\int_{\muI} (-\lambda)^{-1/\lambda}\dsp F_V(v)=(-\lambda)^{-1/\lambda}\upm.
\]
Since this limit is non-zero it follows that $\xmarg$ is slowly varying. The result for $\xq$ now follows from Lemma~\ref{geninv:lem} (with $\ixmarg=\Isvo^{\lambda}$). The $\yq$ case is similar.
\end{proof}

Let $\Delta$ be a neighbourhood of $1/2$ on which $F'_V$ is continuous; in particular, $\dsp F_V(v)=F_V'(v)\dsp v$ for $v\in\Delta$.

\begin{proof}[Proof of Proposition~\ref{l<0xyj:prop}]
Set $r(t)=\xq/\{\xq+\yq\}$ so \eqref{GPjointXY:eq1} gives $\Prob\bigl\{A > \xq,\,B > \yq \bigr\}=\intg{-}+\intg{+}$ where
\[
\intg{-}\!=\!\int_0^{r(t)}\{1+\lambda\xq\nf(v)\}_+^{-1/\lambda}\dsp F_V(v)
\quad\text{and}\quad 
\intg{+}\!=\!\int_{r(t)}^1\{1+\lambda\yq\nf(1-v)\}_+^{-1/\lambda}\dsp F_V(v).
\]
To consider $\intg{-}$ firstly set $v_-(t)=\xq/\{\xq+\esupp\}$. 
Since $\yq < \esupp$ and $\xq,\yq\to\esupp$ as $t\to\infty$ we get $v_-(t)<r(t)$ while $v_-(t),r(t)\to1/2$ as $t\to\infty$.
As $\rep>1/2$ we can then choose $T_0$ so that $[v_-(t),r(t)]\subseteq[1-\rep,\rep]\cap\Delta$ whenever $t\ge T_0$.
For $t\ge T_0$ it follows that $\nf(v)=\max\{(1-v)/v,1\}$ when $v\in[v_-(t),r(t)]$; in particular $\nf\{v_-(t)\}=\esupp/\xq$.
Furthermore \eqref{shapemu:eq} implies $\nf(v)$ is decreasing on $[0,r(t)]$. For $v\in[v_-(t),r(t)]$ we thus have
\[
1+\lambda\xq\nf(v)>0
\ \Longrightarrow\ \nf(v)<\frac{1}{-\lambda\xq}=\nf\{v_-(t)\}
\ \Longrightarrow\ v>v_-(t).
\]
Therefore
\begin{align*}
\intg{-}&=\int_{v_-(t)}^{r(t)}\left\{1+\lambda\xq \max\Bigl(\frac{1-v}{v},\,1\Bigr)\right\}_+^{-1/\lambda}\dsp F_V(v)\\
&=\int_{v_-(t)}^{r(t)}\min\left\{1+\lambda\xq \frac{1-v}{v},\,1+\lambda\xq \right\}_+^{-1/\lambda}\dsp F_V(v).
\end{align*}
Consider the new variable $u=\{1+\lambda\yq\}^{-1}\bigl\{1+\lambda\xq(1-v)/v\bigr\}$, and its inverse
$v=-\lambda\xq[1-\lambda\xq-\{1+\lambda\yq\}u]^{-1}$. We have $u=0$ (respectively $u=1$) when $v=v_-(t)$ (respectively $v=r(t)$). Thus
\begin{subequations}
\label{Int-:eqs}
\begin{align}
\intg{-}
\label{Int-:eq1}
&=\{1+\lambda\yq\}^{1-1/\lambda}\int_0^1\min\left\{u,\,\frac{1+\lambda\xq}{1+\lambda\yq}\right\}^{-1/\lambda} G^-_t(u)\dsp u\\
\label{Int-:eq2}
&=\{1+\lambda\yq\}\,\{1+\lambda\xq\}^{-1/\lambda}\int_0^1\min\left\{\frac{1+\lambda\yq}{1+\lambda\xq}\,u,1\right\}^{-1/\lambda} G^-_t(u)\dsp u,
\end{align}
\end{subequations}
where
\[
G^-_t(u)=F'\left[\frac{-\lambda\xq }{1-\lambda\xq -\{1+\lambda\yq\}u}\right]\,\frac{-\lambda\xq}{\bigl[1-\lambda\xq-\{1+\lambda\yq\}u\bigr]^2}.
\]
As $t\to\infty$ we have $-\lambda\xq,-\lambda\yq \to1$ so $1-\lambda\xq -\{1+\lambda\yq\}u\to2$, uniformly for $u\in[0,1]$. 
Using Assumption~\ref{As_cd} it follows that $G^-_t(u)\to F_V'(1/2)/4=:\Gamma$, uniformly for $u\in[0,1]$. Hence the integrands in \eqref{Int-:eqs} are uniformly bounded for all sufficiently large $t$. Proposition \ref{l<0xmarg:prop} gives
\begin{equation}
\label{xqyqasylim:eq}
1+\lambda\xq=\upm^\lambda t^{\lambda\beta}\{1+o(1)\}
\ \quad\text{and}\quad\ 
1+\lambda\yq=\lowm^\lambda t^{\lambda\gamma}\{1+o(1)\}
\ \text{as $t\to\infty$.}
\end{equation}

\noindent
\emph{If $\beta<\gamma$:}\quad
As $t\to\infty$ we have $\{1+\lambda\xq\}/\{1+\lambda\yq\}\to+\infty$ by \eqref{xqyqasylim:eq}, so applying dominated convergence to \eqref{Int-:eq2} gives
\[
\intg{-}=\lowm^{\lambda-1} t^{\lambda\gamma-\gamma}\int_0^1u^{-1/\lambda}\Gamma\dsp u\;\{1+o(1)\}
=O(t^{\lambda\gamma-\gamma})=o(t^{\lambda\beta-\gamma}).
\]

\noindent
\emph{If $\beta>\gamma$:}\quad
As $t\to\infty$ we have $\{1+\lambda\yq\}/\{1+\lambda\xq\}\to+\infty$ by \eqref{xqyqasylim:eq}, so applying dominated convergence to \eqref{Int-:eq2} gives
\[
\intg{-}=\lowm^{\lambda}t^{\lambda\gamma}\upm^{-1}t^{-\beta}\int_0^11^{-1/\lambda}\Gamma\dsp u\;\{1+o(1)\}
=\Gamma\lowm^{\lambda}\upm^{-1}\,t^{\lambda\gamma-\beta}\{1+o(1)\}.
\]

\noindent
\emph{If $\beta=\gamma$:}\quad
As $t\to\infty$ we have $\{1+\lambda\xq\}/\{1+\lambda\yq\}\to\upm^{\lambda}/\lowm^{\lambda}$ by \eqref{xqyqasylim:eq}, so applying dominated convergence to \eqref{Int-:eq1} gives
\[
\intg{-}=\lowm^{\lambda-1} t^{\lambda\beta-\beta}\int_0^1\min\left(u,\,\frac{\upm^{\lambda}}{\lowm^{\lambda}}\right)^{-1/\lambda}\Gamma\dsp u\;\{1+o(1)\}.
\]
When $\upm\le\lowm$ this becomes $\intg{-}=-\Gamma \lambda(1-\lambda)^{-1}\lowm^{\lambda-1}\,t^{\lambda\beta-\beta}\{1+o(1)\}.$
When $\upm\ge\lowm$ we get
\begin{align*}
\intg{-}
&=\Gamma\lowm^{\lambda-1}\left(\int_0^{\frac{\upm^\lambda}{\lowm^\lambda}}u^{-1/\lambda}\dsp u
+\int_{\frac{\upm^\lambda}{\lowm^\lambda}}^1\frac{\lowm}{\upm}\dsp u\right)t^{\lambda\beta-\beta}\{1+o(1)\}=\Gamma\left(\lowm^\lambda-\frac{\upm^\lambda}{1-\lambda}\right)\upm^{-1}\,t^{\lambda\beta-\beta}\{1+o(1)\}.
\end{align*}
\noindent
A similar calculation for $\intg{+}$ leads to 
\[
\intg{+}
=\begin{cases}
\Gamma\upm^\lambda\lowm^{-1}\,t^{\lambda\beta-\gamma}\{1+o(1)\}&\text{if $\beta<\gamma$,}\\
o(t^{\lambda\gamma-\beta})&\text{if $\beta>\gamma$,}\\  
\Gamma\bigl(\upm^\lambda-\frac1{1-\lambda}\,\lowm^\lambda\bigr)\lowm^{-1}\,t^{\lambda\beta-\beta}\{1+o(1)\}&\text{if $\beta=\gamma$ and $\upm\le\lowm$,}\\[2pt]
\Gamma\,\frac{-\lambda}{1-\lambda}\,\upm^{\lambda-1}\,t^{\lambda\beta-\beta}\{1+o(1)\}&\text{if $\beta=\gamma$ and $\upm\ge\lowm$,}    
\end{cases}
\]
as $t\to\infty$.
This is combined with $\intg{-}$ to give \eqref{xyjlim:eq}. As the limit is non-zero in all cases, $\xyj$ is slowly varying.
\end{proof}

\section{Derivations of ray dependence functions ($\lambda>0$ and $\lambda<0$) and spectral density ($\lambda>0$)}
\label{sec:AppB}

\subsubsection*{Derivation of $d(q)$ for $\lambda>0$}

This follows simply by noting that Proposition~6 gives that marginal quantile functions are
\[
 q_A(tx)= (tx)^{\lambda}\ixmarg(tx),~~
 q_B(ty)= (ty)^{\lambda}\iymarg(ty),
\]
 for $tx, ty \ge 1$ so that using the same dominated convergence arguments as in $\lim_{t\to\infty}\theta(t)$ given in the proof of Proposition~1,
\begin{align}
\label{eq:dlim}
 \lim_{t\to\infty} t\Prob\{A>q_A(tx), B>q_B(ty)\} = \lambda^{-1/\lambda} \int_0^1 \min\left\{\frac{\nf(v)^{-1/\lambda}}{\upn x},\frac{\nf(1-v)^{-1/\lambda}}{\lown y}\right\}\, \dsp F_V(v).
\end{align}
Therefore $\Prob\{A>q_A(tq), B>q_B(t(1-q))\}/\Prob\{A>q_A(t), B>q_B(t)\}$ converges to $q^{-1/2}(1-q)^{-1/2} d(q)$ with $d$ the form claimed in Remark~1.\\ \vspace{0.1cm}

\subsubsection*{Derivation of $h$ for $\lambda>0$}

To derive $h$, consider~\eqref{eq:dlim}, with $\dsp F_V(v) = f_V(v) \dsp v$. This expression can be set equal to
\begin{align*}
\int_0^1 2\min\left(\frac{w^*}{x},\frac{1-w^*}{y}\right) h(w^*) \dsp w^* = \int_0^{\frac{x}{x+y}} \frac{2w^*}{x} h(w^*) \dsp w^* +  \int_{\frac{x}{x+y}}^1 \frac{2(1-w^*)}{x} h(w^*) \dsp w^*.
\end{align*}
By differentiating under the integral sign, we have
\begin{align*}
\frac{\partial^2}{\partial x\partial y} \left\{\int_0^{\frac{x}{x+y}} \frac{2w^*}{x} h(w^*) \dsp w^* +  \int_{\frac{x}{x+y}}^1 \frac{2(1-w^*)}{y} h(w^*) \dsp w^*\right\} = \frac{2}{(x+y)^3} h\left(\frac{x}{x+y}\right),
\end{align*}
so that $h$ is recovered upon setting $x=w,y=1-w$, and dividing by two. Thus we begin with
\begin{align*}
\lambda^{-1/\lambda}  \int_0^1 \min\left\{\frac{\nf(v)^{-1/\lambda}}{\upn x}, \frac{\nf(1-v)^{-1/\lambda}}{\lown y}\right\} f_V(v) \dsp v &= \lambda^{-1/\lambda} \int_0^{r(x,y)} \frac{\nf(v)^{-1/\lambda}}{\upn x} f_V(v) \dsp v \\& ~ + \lambda^{-1/\lambda} \int_{r(x,y)}^1 \frac{\nf(1-v)^{-1/\lambda}}{\lown y} f_V(v) \dsp v,
\end{align*}
with $r(x,y) = \frac{(x\upn)^\lambda}{(x\upn)^\lambda+(y\lown)^\lambda}$. Differentiating with respect to $x$ yields
\begin{align*}
 \lambda^{-1/\lambda} \Bigg\{\int_0^{r(x,y)} -\frac{\nf(v)^{-1/\lambda}}{\upn x^2} f_V(v) \dsp v + \frac{\nf\{r(x,y)\}^{-1/\lambda}}{\upn x}f_V\{r(x,y)\}\frac{\partial}{\partial x} r(x,y)& \\- \frac{\nf\{1-r(x,y)\}^{-1/\lambda}}{\lown y}f_V\{r(x,y)\}\frac{\partial}{\partial x} r(x,y)\Bigg\} 
 &=  \int_0^{r(x,y)} -\frac{\nf(v)^{-1/\lambda}}{\upn x^2} f_V(v) \dsp v,
\end{align*}
whilst differentiating what remains with respect to $y$ gives
\begin{align*}
 -\lambda^{-1/\lambda} \frac{\nf\{r(x,y)\}^{-1/\lambda}}{\upn x^2}f_V\{r(x,y)\}\frac{\partial}{\partial y} r(x,y).
\end{align*}
Substituting in $\nf$ and noting that 
\[
 \frac{\partial}{\partial y} r(x,y) =  \frac{\partial}{\partial y} \frac{(x\upn)^\lambda}{(x\upn)^\lambda+(y\lown)^\lambda} = -\lambda \frac{x^\lambda y^{\lambda-1} \upn^\lambda \lown^\lambda}{\{(x\upn)^\lambda+(y\lown)^\lambda\}^2}
\]
gives
\begin{align*}
\frac{x^{\lambda-1}y^{\lambda-1} \upn^\lambda \lown^\lambda}{\|(x\upn)^\lambda,(y\lown)^\lambda\|_m^{1/\lambda}\{(x\upn)^\lambda+(y\lown)^\lambda\}^2} f_V\left\{\frac{(x\upn)^\lambda}{(y\upn)^\lambda+(y\lown)^\lambda}\right\},
\end{align*}
so that substituting $x=w,y=1-w$ and dividing by two yields
\begin{align*}
 h(w)=\frac{\lambda^{1-1/\lambda}}{2}\frac{w^{\lambda-1}(1-w)^{\lambda-1} \upn^\lambda \lown^\lambda}{\|(w\upn)^\lambda,((1-w)\lown)^\lambda\|_m^{1/\lambda}\{(w\upn)^\lambda+((1-w)\lown)^\lambda\}^2} f_V\left\{\frac{(w\upn)^\lambda}{(w\upn)^\lambda+((1-w)\lown)^\lambda}\right\},\\ \vspace{0.1cm}
\end{align*}
which is denoted $h(\cdot;\lambda,f_V)$ in Remark~\ref{rmk:1}.

\subsubsection*{Derivation of $d(q)$ for $\lambda<0$}

This follows firstly by noting that Proposition~9 gives that marginal quantile functions are
\[
 q_A(tx)= \Lambda - (tx)^{\lambda}\ixmarg(tx),~~
 q_B(ty)= \Lambda - (ty)^{\lambda}\iymarg(ty),
\]
 for $tx, ty \ge 1$. The ray dependence function can be found by following the proof of Proposition~4 through with these $q_A(tx)$ and $q_B(ty)$, which reveals that 
\begin{align*}
\label{eq:dlim2}
 \lim_{t\to\infty} t^{1-\lambda} \Prob\{A>q_A(tx), B>q_B(ty)\} =\frac{F'_V(1/2)}{4} \bigl\{\min(x\upm,y\lowm)^\lambda-\frac{1+\lambda}{1-\lambda}\max(x\upm,y\lowm)^{\lambda}\bigr\}\max(x\upm,y\lowm)^{-1}.
\end{align*}
Therefore $\Prob\{A>q_A(tq), B>q_B(t(1-q))\}/\Prob\{A>q_A(t), B>q_B(t)\}$ converges to $q^{-\frac{1-\lambda}{2}}(1-q)^{-\frac{1-\lambda}{2}} d(q)$ with $d$ the form claimed in Remark~2.

\bibliographystyle{apalike}
\bibliography{ADAIBib}

\end{document}